\documentclass[a4paper,11pt,reqno]{amsart}
\usepackage[utf8]{inputenc}
\usepackage{mathrsfs,dsfont,hyperref,amsmath,amssymb,amsthm,amsfonts,amstext,amsopn,amsxtra,mathrsfs,esint,enumitem,bookmark,mathtools}
\usepackage[capitalize]{cleveref}
\usepackage{braket}
\usepackage{xcolor}
\usepackage{mathabx}

\newtheorem{theorem}{Theorem}%[section]
\newtheorem{proposition}[theorem]{Proposition}
\newtheorem{lemma}[theorem]{Lemma}

\theoremstyle{definition}

\theoremstyle{remark}

\def\bR{\mathbb{R}}
\def\bZ{\mathbb{Z}}
\def\cF{\mathcal{F}}
\def\cB{\mathcal{B}}
\def\cN{\mathcal{N}}

\newcommand{\E}{\mathcal{E}}

\newcommand{\N}{\mathcal{N}}

\newcommand{\R}{\mathbb{R}}

\newcommand\1{{\ensuremath {\mathds 1} }}
\newcommand{\dt}{{\rm d}t}
\newcommand{\ds}{{\rm d}s}
\newcommand{\dx}{{\rm d}x}
\newcommand{\dy}{{\rm d}y}
\newcommand{\dz}{{\rm d}z}
\newcommand{\du}{{\rm d}u}
\newcommand{\dv}{{\rm d}v}

\newcommand{\B}{{\mathcal{B}}}

\newcommand{\vphi}{{\varphi}}
\newcommand{\eps}{\varepsilon}
\newcommand{\ca}{{\check a}}
\newcommand{\aN}{\mathfrak{a}_N}

\DeclareMathOperator{\tr}{Tr}

\author[C. Hainzl]{Christian Hainzl}
\address{Department of Mathematics, LMU Munich, Theresienstrasse 39, 80333 Munich, and Munich Center for Quantum Science and Technology, Schellingstr. 4, 80799 Munich, Germany}  
\email{hainzl@math.lmu.de}

\author[B. Schlein]{Benjamin Schlein}
\address{Institute of Mathematics, University of Zurich, Winterthurerstrasse 190, 8057 Zurich, Switzerland}  
\email{benjamin.schlein@math.uzh.ch}

\author[A. Triay]{Arnaud Triay}
\address{Department of Mathematics, LMU Munich, Theresienstrasse 39, 80333 Munich, and Munich Center for Quantum Science and Technology, Schellingstr. 4, 80799 Munich, Germany}  
\email{triay@math.lmu.de}

\begin{document}
\title[Bogoliubov theory]{Bogoliubov Theory in the Gross-Pitaevskii Limit: \\ a Simplified Approach}

\begin{abstract}
We show that Bogoliubov theory correctly predicts the low-energy spectral properties of Bose gases in the Gross-Pitaevskii regime. We recover recent results from \cite{BBCS,BBCS4}. While our main strategy is similar to the one developed in  \cite{BBCS,BBCS4}, we combine it with new ideas, taken in part from \cite{Hainzl,NT}; this makes our proof substantially simpler and shorter. As an important step towards the proof of Bogoliubov theory, we show that low-energy states exhibit complete Bose-Einstein condensation with optimal control on the number of orthogonal excitations.  
\end{abstract}
\maketitle
\section{Introduction}

We consider a Bose gas consisting of $N$ particles moving in the box $\Lambda = [-1/2 ; 1/2]^3$, with periodic boundary conditions. In the Gross-Pitaevskii regime, particles interact through a potential with scattering length of the order $1/N$. The Hamilton operator acts on the Hilbert space $L^2_s (\Lambda^N)$ of permutation symmetric complex valued square integrable functions on $\Lambda^N$ and it has the form 
\begin{equation} \label{eq:ham} 
H_N = \sum_{i=1}^N -\Delta_i +  \sum_{i<j}^N V_N(x_i-x_j) 
\end{equation} 
where 
 $$V_N(x) :=  N^{2} V(N x),$$ for a $V\in L^{2}(\R^3)$ non-negative, radial and compactly supported. We denote the scattering length of $V$ by $\mathfrak{a} > 0$. Following \cite{Hainzl,HS} we define it through the formula 
 \begin{equation}\label{sclBorn} 4\pi \mathfrak{a} = \frac{1}{2} \int_{\R^3} V(x) \dx  - \left \langle \frac  12 V, \frac{1} {-\Delta + \frac 1 2 V} \frac 1 2 V \right \rangle.
\end{equation}

As first proven in \cite{LY,LSY}, the ground state energy $E_N$ of (\ref{eq:ham}) satisfies 
\begin{equation}\label{eq:LY} E_N / N \to 4\pi \mathfrak{a} \end{equation} 
in the limit $N \to \infty$. In particular, to leading order, the ground state energy only depends on the interaction potential through its scattering length $\mathfrak{a}$. In \cite{LieSei-02,LS2,NamRouSei-15}, it was also shown that the ground state of (\ref{eq:ham}) and, in fact, every normalized sequence $\psi_N \in L_s^2 (\Lambda^N)$ of approximate ground states, with  
\[ \frac{1}{N} \, \langle \psi_N , H_N \psi_N \rangle \to 4\pi \mathfrak{a} , \]
exhibits complete Bose-Einstein condensation in the zero-momentum state $\phi_0 (x) = 1$, for all $x \in \Lambda$, in the sense that the corresponding one-particle reduced density matrix $\gamma_N$ (normalized so, that $\tr \, \gamma_N = 1$) satisfies 
\[ \lim_{N \to \infty}  \langle \phi_0, \gamma_N \phi_0 \rangle = 1 \, . \]

Recently, a rigorous version of Bogoliubov theory \cite{bog} has been developed in \cite{BBCS0,BBCS1,BBCS,BBCS4} to provide more precise information on the low-energy spectrum of (\ref{eq:ham}), resolving the ground state energy and the low-lying excitations up to errors that vanish in the limit $N \to \infty$, and on the corresponding eigenvectors, showing Bose-Einstein condensation with optimal control on the number of orthogonal excitations. Analogous results have been established also for Bose gases trapped by external potentials in the Gross-Pitaevskii regime \cite{NNRT,BSS1,NT,BSS2} and for Bose gases in scaling limits interpolating between the Gross-Pitaevskii regime and the thermodynamic limit \cite{ABS,BCaS}. Very recently, the upper bound for the ground state energy has been also extended to the case of hard-sphere interaction, as announced in \cite{BCOPS}. 

In this paper, we propose a new and substantially simpler proof of the results established in \cite{BBCS,BBCS4}. Our approach follows some of the ideas in the proof of Bose-Einstein condensation with optimal bounds on the number of excitations obtained in \cite{Hainzl}. Moreover, it makes use of some ideas introduced in \cite{NT}, for the case of particles trapped by an external potential. The next theorem is our main result, it describes the low-energy spectrum of (\ref{eq:ham}). 
\begin{theorem} \label{thm:main} 
Let $V \in L^{2} (\bR^3)$ be non-negative, radial and compactly supported, and let $E_N$ denote the ground state energy of (\ref{eq:ham}). 
Then, the spectrum of $H_N - E_N$ below a threshold $\Theta \leq N^{1/17}$ consists of eigenvalues having the form 
\begin{equation}\label{eq:excit} \sum_{p \in 2\pi \bZ^3 \backslash \{ 0 \}} n_p \sqrt{|p|^4 + 16 \pi \mathfrak{a} p^2} + \mathcal O (N^{-1/17} \Theta) \end{equation} 
with $n_p \in \mathbb{N}$, for all $p \in 2\pi \mathbb{Z}^3 \backslash \{ 0 \}$. 
\end{theorem} 

{\it Remark.} Our analysis also provides a precise estimate for the ground state energy $E_N$ of (\ref{eq:ham}), showing that 
\begin{equation}\label{eq:ENbd} E_N = 4\pi \mathfrak{a}_N (N-1) + \frac{1}{2} \sum_{p \in 2\pi\mathbb{Z}^3\backslash \{ 0 \}} \Big[ \sqrt{|p|^4 + 16 \pi \mathfrak{a} p^2} - p^2 - 8\pi \mathfrak{a} + \frac{(8\pi \mathfrak{a})^2}{2p^2} \Big] + \mathcal O (N^{-1/17}) \end{equation} 
with a ``box scattering length'' $\mathfrak{a}_N$, which will be defined in the next section, satisfying $|\mathfrak{a}_N - \mathfrak{a}| \lesssim N^{-1}$. This immediately implies that $E_N$ is given by (\ref{eq:ENbd}), with $\mathfrak{a}_N$ replaced by the true scattering length $\mathfrak{a}$, up to an error that remains bounded, as $N \to \infty$. In \cite{BBCS}, also the order one correction arising from $N (\mathfrak{a}_N - \mathfrak{a})$ was computed. Here we skip this step, to keep our presentation as simple as possible. Note that an estimate similar to (\ref{eq:ENbd}) has been recently shown to hold in the thermodynamics limit; see \cite{FS,FS2} for the lower bound and \cite{YY,BCS} for the upper bound. 

The main strategy we use to prove Theorem \ref{thm:main} is similar to the one developed in \cite{BBCS}. First of all, we switch to the formalism of second quantization, expressing the Hamilton operator in momentum space, in terms of creation and annihilation operators. Afterwards, we renormalize the Hamilton operator, conjugating it first with a generalized Bogoliubov transformation (the exponential of a quadratic expression in the modified creation and annihilation operators $b_p^\dagger = a_p^\dagger a_0 / \sqrt{N}$, $b_p = a_0^\dagger a_p / \sqrt{N}$) and afterwards with the exponential of a cubic expression in (modified) creation and annihilation operators. Effectively, these conjugations regularize the interaction potential. At the end, we diagonalize the resulting quadratic Hamiltonian; this allows us to establish Bose-Einstein condensation with optimal bounds on the number of excitations and to compute the low-energy spectrum, proving (\ref{eq:excit}). 

Compared with \cite{BBCS}, our approach has the following advantages. First of all, we make a different choice for the coefficients $\vphi_p$ of the quadratic and cubic transformations that are used to renormalize the Hamiltonian and that should model correlations among particles.  Instead of the ground state of a Neumann problem on a ball of radius $\ell >0$, we consider here the solution of an appropriate zero-energy scattering equation, describing scattering processes inside the box $\Lambda$. This simplifies the proof of important properties of $\vphi$ and improves cancellations between different terms arising in the many-body analysis. Second, we restrict the quadratic conjugation to momenta $|p| > N^\alpha$, for some $0 < \alpha <1$. As a consequence, it is enough to expand its action to first or, in few cases, to second order; higher order contributions are negligible. This is a substantial advantage, compared with \cite{BBCS}, where no cutoff was imposed and all contributions had to be computed precisely (in contrast to standard Bogoliubov transformations, the action of generalized Bogoliubov transformations is not explicit). The presence of the cutoff means that the interaction is regularized only up to length scales $\ell \leq N^{-\alpha}$; this needs to be compensated at the end, when we diagonalize the quadratic Hamiltonian resulting from the renormalization procedure. Another important simplification of the analysis concerns the final diagonalization. As in \cite{BBCS}, we implement it through a generalized Bogoliubov transformation, defined (like the first quadratic transformation) in terms of the modified creation and annihilation operators $b_p^\dagger, b_p$. Here, however, instead of expanding the action of the generalized Bogoliubov transformation to all orders, we compare it directly with the explicit action of the corresponding standard Bogoliubov transformation, making use of an appropriate interpolation. Finally, we use the tool of localization in the number of particles, not only to show Bose-Einstein condensation (similarly as in \cite{BBCS4}), but also to compute the spectrum and prove Theorem \ref{thm:main}, This makes the analysis substantially simpler (but it provides a worse estimate on the error).

\section{Fock space formalism} 

We introduce the bosonic Fock space 
\[ \cF = \bigoplus_{n \geq 0} L^2_s (\Lambda^n). \]
For a momentum $p \in \Lambda^\ast = 2\pi \mathbb{Z}^3$ and denoting $u_p(x) = e^{i p \cdot x}$ we define $a_p^\dagger = a^\dagger(u_p)$ and $a_p = a(u_p)$, where $a^\dagger$ and $a$ are the usual creation and annihilation operators. They satisfy the canonical commutation relations 
\begin{equation}\label{eq:CCR} \big[ a_p , a_q^\dagger \big] = \delta_{p,q}, \qquad \big[ a_p, a_q \big] = \big[ a_p^\dagger , a_q^\dagger \big] = 0. \end{equation} 
We denote, in configuration space, the creation and annihilation operator valued distributions by $\ca^\dagger_x, \ca_x$, they satisfy $a^\dagger_p = \int e^{ip\cdot x} \ca^\dagger_x \dx$, $a_p = \int e^{-ip\cdot x} \ca_x \dx$. The number of particles operator $\cN$ on $\cF$ is given by 
\[ \cN = \sum_{p \in \Lambda^\ast} a_p^\dagger a_p. \]

In the formalism of second quantization, the Hamilton operator (\ref{eq:ham}) takes the form 
\begin{equation}\label{eq:HN-fock}
H_N = \sum_{p \in \Lambda^\ast} p^2  a_p^\dagger a_p + \frac{1}{2}\sum_{r,p,q \in \Lambda^\ast} \hat{V}_N (r) a_{p+r}^\dagger a_q^\dagger a_p a_{q+r},
\end{equation} 
with 
\begin{equation} \label{eq:hatVN} \hat{V}_N (r) = \frac 1{N} \hat V(r/N).\end{equation}  
To recover (\ref{eq:ham}), we have to restrict (\ref{eq:HN-fock}) on the sector with $\cN = N$. 

Because of the presence of Bose-Einstein condensation, the mode with $p=0$ plays a special role when considering states with low energy. We introduce the notation $\cN_0 = a_0^\dagger a_0$  and $\cN_+ = \cN - \cN_0$ for the operators measuring the number of particles in the condensate and the number of excitations, respectively. Following Bogoliubov \cite{bog}, we decompose (\ref{eq:HN-fock}) according to the number of $a_0, a_0^\dagger$ operators. 
Since (on $\{ \cN = N \}$) 
%
%Let $\phi_p(x) = e^{ixp}, p \in \Lambda^\ast :=  2\pi \mathbb{Z}^3$. This gives an ONB of $L^2(\Lambda)$.
%The $\phi_p$ are eigenfunctions with periodic boundary conditions of $-\Delta$ with eigenvalues $p^2$. \\
%Let $a^\dagger_p := a^\dagger(\phi_p)$ and $a_p :=  a(\phi_p)$. From \eqref{CCR} we immediately see
%\[ \label{CCRdiscrete}
%[a^\dagger_p, a^\dagger_q] = 0 = [a_p, a_q], \quad [a_p, a^\dagger_q] = \delta_{p,q}.
%\]
%
 $$ a^\dagger_0 a^\dagger_0 a_0 a_0 = \N_0(\N_0 - 1) = (N - \N_+)(N - \N_+ -1) = N(N - 1) - \N_+(2N - 1) + \N_+^2,$$
we can rewrite (\ref{eq:HN-fock}) as 
\begin{equation}
	\label{eq:decompo_HN}
H_N = H_0+ H_1 + H_2 + Q_2 + Q_3 + Q_4,
\end{equation}
where
\[ \begin{split} H_0 &= \frac{\hat{V}_N(0) }{2}  N (N - 1 ) , \qquad H_1 = \sum_{p \neq 0} p^2  a_p^\dagger a_p, \\ H_2 &= \sum_{p\neq 0} \hat V_N(p) a^\dagger_p a_p(N - \N_+)  - \frac {\hat V_N(0) }{2} \N_+(\N_+ -1). 
\end{split} \]
and 
\begin{equation}
\label{eq:Qs} 
\begin{split} 
Q_2 = & \frac{1}{2}\sum_{p \neq 0} \hat V_N(p) [a_p^\dagger a_{-p}^\dagger a_0 a_0 + \text{h.c.}], \\
Q_3 = &  \sum_{q,r,q+r\neq 0} \hat V_N(r)\left[a_{q+r}^\dagger a_{-r}^\dagger a_{q}a_0 + \text{h.c.}  \right],   \\
Q_4 = & \frac{1}{2}\sum_{p,q \neq 0, r \neq -p, r\neq -q} \hat V_N(r) a_{p+r}^\dagger a_q^\dagger a_p a_{q+r}.
\end{split}
\end{equation}

Since we isolated the contributions of the zero modes, we follow from now on the convention that the indices appearing in creation and annihilation operators are always non-zero,  except stated otherwise.

Naive power counting, based on the fact that $a_0 , a_0^\dagger \simeq \sqrt{N}$ due to the presence of Bose-Einstein condensation and on the scaling (\ref{eq:hatVN}) of the interaction, suggests that the terms $Q_3$ and $Q_4$ are small. For this reason, Bogoliubov neglected these contributions, and diagonalized the remaining quadratic terms. This led to expressions similar to (\ref{eq:excit}), (\ref{eq:ENbd}) for the low-energy spectrum of (\ref{eq:ham}), but with the scattering length replaced by its first and second Born approximations. In fact, because of the slow decay of the potential in Fourier space, the operators $Q_3$ and $Q_4$ are not small. They contain instead important terms, which effectively renormalize the interaction and produce the scattering length appearing in the formulas (\ref{eq:excit}), (\ref{eq:ENbd}). To obtain a rigorous proof of Theorem 1, it is therefore crucial that we first extract the large contributions to the energy that are hidden in the cubic and quartic operators $Q_3$, $Q_4$; only afterwards we can diagonalize the remaining quadratic terms.  

Let us give a little more detail about the main ideas of the proof. Following the strategy of \cite{BBCS}, we will first conjugate (\ref{eq:decompo_HN}) with a unitary operator of the form $e^{\cB_2}$, where $\cB_2$ is a quadratic expression in creation and annihilation operators $a_p, a_p^\dagger$, associated with momenta $p \not = 0$. The goal of this conjugation is to extract contributions that regularize the off-diagonal term $Q_2$ and, at the same time, reconstruct the leading order ground state energy $4 \pi {\mathfrak a}_N N$, when combined with $H_0$. Roughly speaking, neglecting several error terms, we will find  
\begin{equation}\label{eq:eBHeB}
\begin{split}  e^{-\B_2} H_N e^{\B_2}  \simeq  & \, \, 4 \pi \aN (N-1) + \sum_{|p|\leq N^{\alpha}} \frac{(4\pi \aN)^2}{p^2}  \\ &+ \sum_{p \in \Lambda^*} (p^2 + 2\hat V(0) - 8\pi \aN) a^\dagger_p a_p  + \sum_{|p| \leq N^\alpha} 4\pi \mathfrak{a}_N[ a_p^\dagger a_{-p}^\dagger + \text{h.c.} ] + Q_3 + Q_4.
\end{split}
\end{equation} 
As explained in the introduction, an important difference, compared with \cite{BBCS}, is that here we impose an infrared cutoff in $\cB_2$, defined in (\ref{defB}), letting it act only on momenta $|p| > N^\alpha$. On the one hand, this choice simplifies the computation of the action of $\cB_2$ (it allows us to expand it; important contributions arise only from the first and second commutator). On the other hand, it produces terms, like the sum on the first line and the regularized off-diagonal quadratic term on the second line of (\ref{eq:eBHeB}), which contribute to the energy to order $N^\alpha$; these terms are larger than the precision we are looking for, and will need to be compensated for, with the second quadratic transformation. Notice that the idea of using an infrared cutoff in the quadratic conjugation already appeared in the proof of complete Bose-Einstein condensation given in \cite{BBCS0} and, more recently, in the proof of the validity of Bogoliubov theory for Bose gases trapped by an external potential obtained in \cite{NT}. 

Observing (\ref{eq:eBHeB}), it is clear that we still have to renormalize the diagonal quadratic term (proportional to $\hat{V} (0)$) and the cubic term $Q_3$. To this end, we will introduce a unitary transformation $e^{\cB_3}$, with $\cB_3$, defined in (\ref{defB3}), cubic in the operator $a_p, a_p^\dagger$, with $p \not = 0$. Up to several negligible errors, conjugation with $e^{\cB_3}$ will lead us to 
\begin{equation}\label{eq:eB3HeB3} \begin{split} e^{-\cB_3} e^{-\B_2} H_N e^{\B_2}e^{\cB_3} \simeq \, & \, \, 4 \pi \aN (N-1) + \sum_{|p|\leq N^{\alpha}} \frac{(4\pi \aN)^2}{p^2}.
 \\ &+ \sum_{p \in \Lambda^*} (p^2 +  8\pi \aN) a^\dagger_p a_p  + \sum_{|p| \leq N^\alpha} 4\pi \mathfrak{a}_N[ a_p^\dagger a_{-p}^\dagger + \text{h.c.} ] + Q_4 \, . 
 \end{split}
 \end{equation} 
The only term on the r.h.s. of the last equation, where we still have the original, singular, potential $\hat{V}_N$ is $Q_4$; all other terms have been renormalized and are now expressed in terms of the scattering length $\frak{a}_N$. Fortunately, $Q_4$ is positive; for this reason, we do not need to renormalize it (for lower bounds, it can be neglected; for upper bounds, it only needs to be controlled on special trial states). Finally, in section \ref{sec:B4},  we will apply a second quadratic transformation $e^{\cB_4}$ to diagonalize the remaining quadratic Hamiltonian on the r.h.s. of (\ref{eq:eB3HeB3}). This will lead us to 
\begin{align*} 
e^{-\cB_4}e^{-\cB_3} e^{-\B_2} H_N e^{\B_2}e^{\cB_3} e^{\cB_4} \simeq 
  &\; 4 \pi \aN (N-1)  + \frac{1}{2}\sum_{p} \left[ \sqrt{p^4 + 16 \pi \aN p^2} - p^2 - 8 \pi \aN + \frac{(8\pi \aN)^2}{2p^2}\right] \\ &+ \sum_{p} \sqrt{p^4 + p^2 16 \pi \aN } \, a^\dagger_p a_p + Q_4
\end{align*}
which will allow us to show Theorem \ref{thm:main}. To control error terms, we use the tool of localization in the number of particles to show Bose-Einstein condensation (similarly  to \cite{BBCS4,NT}).

\section{Quadratic renormalization} 
\label{sec:quadra}

Starting with the quadratic transformation we conjugate (\ref{eq:HN-fock}) with the unitary $e^{\cB_2}$, where 
\begin{equation}\label{defB}
 \B_2 = \frac 1{2} \sum_{p} \tilde \vphi_p [ a^\dagger_p a^\dagger_{-p} a_0 a_0 - \rm{h.c.}].
\end{equation}
We are going to fix the coefficients $\tilde\vphi_p$ so that the commutator $[ H_1 + Q_4 , \cB_2 ]$ arising from the action of (\ref{defB}) renormalizes the off-diagonal quadratic term $Q_2$ (effectively replacing the singular potential $V_N$ with a regularized interaction having the same scattering length). To this end, we choose $\vphi_p$ satisfying the relations
\begin{equation} \label{scattequ}
p^2 \vphi_p + \frac{1}{2}\sum_{q\neq 0} \hat{V}_N \left((p-q)\right)  \vphi_q = -\frac{1}{2} \hat{V}_N(p) \, ,
\end{equation}
for all $p \in \Lambda^\ast_+ = \Lambda^\ast \backslash \{ 0 \}$. Eq. (\ref{scattequ}) is a truncated version of the zero-energy scattering equation for the potential $V_N$ on the whole space $\mathbb{R}^{3}$.

To prove the existence of a solution of (\ref{scattequ}), we consider the operator \[ \mathfrak{h} = -\Delta + \frac{1}{2} V_N \]
acting on the one-particle space $L^2 (\Lambda)$ (for $N$ large enough, $V_N$ is supported in $[-1/2 ; 1/2]^3$ and can be periodically extended to define a function on the torus). Denoting by $P_0^\perp$ the orthogonal projection onto the orthogonal complement 
of the zero-momentum mode $\vphi_0$ in $L^2 (\Lambda)$, we find (since $V_N \geq 0$), that $P_0^\perp \mathfrak{h} P_0^\perp \geq C > 0$ and therefore that $P_0^\perp \mathfrak{h} P_0^\perp$ is invertible. Thus, we can define $\check\vphi \in L^2 (\Lambda)$ through 
\begin{equation}\label{eq:defphi}  \check\vphi = -\frac{1}{2} P_0^\perp \left[   P_0^\perp \Big(-\Delta + \frac{1}{2} V_N \Big) P_0^\perp \right]^{-1} P_0^\perp V_N\, . \end{equation} 
It is then easy to check that the Fourier coefficients of $\check{\vphi}$ satisfy the relations (\ref{scattequ}). 

Using the sequence $\{ \vphi_p \}_{p \in 2\pi \mathbb{Z}^3 \backslash \{ 0 \}}$, we can define the ``box scattering length'' of $V_N$ by 
\begin{equation}\label{eq:aNdef}
8 \pi \aN := N \Big[ \hat V_N (0) + \sum_p \hat{V}_N (p) \vphi_p \Big] = \hat{V} (0) + N \sum_p \hat{V}_N (p) \vphi_p \, .
\end{equation}
As proven in \cite{Hainzl}, we have that $|\aN - \mathfrak{a}| \lesssim N^{-1}$.

As explained earlier, we renormalize first the high-momenta part of $Q_2$, for this reason we use a cutoff version of $\vphi_p$ to momenta $|p| > N^\alpha$, for some $0< \alpha < 1$. We define therefore 
\begin{equation}\label{eq:tvphi} \tilde \vphi_p = \vphi_p \chi_{|p| > N^\alpha} \, .\end{equation} 
The next lemma lists some important properties of the sequences $\vphi, \tilde{\vphi}$ and of the scattering length $\mathfrak{a}_N$, which will be useful for our analysis.
\begin{lemma}
\label{Lemmavphi} 
Let $V \in L^{2} (\mathbb{R}^3)$ be non-negative and compactly supported. Define $\check\vphi$ as in (\ref{eq:defphi}) and denote by $\vphi_p$ the corresponding Fourier coefficients. Then $\vphi_p \in \bR$, $\vphi_{-p} = \vphi_p$ and
\begin{equation}\label{eq:point} |\vphi_p| \lesssim \frac{1}{N p^2} \end{equation}  
for all $p \in 2\pi \mathbb{Z} \backslash \{ 0 \}$. Moreover, with (\ref{eq:tvphi}), we have 
\begin{align*}
& \|\tilde \vphi \|_2  \lesssim  N^{ - 1 - \alpha/2},\qquad \|\tilde \vphi\|_\infty  \lesssim N^{- 1 - 2\alpha},  \qquad  \|\tilde \vphi \|_1  \lesssim  1,
\end{align*}
%and for all $\xi,\xi' \in \F$, we have
%\begin{align}\label{sumtildephi} 
%	|\langle \xi', \Phi \xi \rangle|
%	&\lesssim N^{-1-\alpha/2}  \| (\N_+ +1)^{1/2} \xi'\|\| (\N_+ +1)^{1/2} \xi\|.
%\end{align}
and 
\begin{align}
N \sum_{p} \hat{V}_N(p) \tilde \vphi_p &= 8\pi \aN - \hat V(0) + \mathcal O(N^{\alpha-1}). \label{eq:approx_aN}
\end{align}
\end{lemma}
\begin{proof}
Multiplying equation \eqref{scattequ} by $\vphi_p$, summing over $p$ and using that $V \geq 0$, we obtain
\begin{equation}\label{eq:pvphi}
2 \|p\vphi\|_2^2  = -  \sum_{p} \hat V_N(p) \vphi_p  -  \sum_{p,q} \hat V_N(p-q) \vphi_p \vphi_q \leq  -  \sum_{p} \hat V_N(p) \vphi_p \, .
\end{equation}
On the one hand, this implies that $\| p \vphi \|_2 \lesssim \| V_N \|_2 \|  \vphi \|_2 < \infty$ (the last bound is not uniform in $N$; it follows from Eq. (\ref{eq:defphi})). On the other hand, (\ref{eq:pvphi}) leads to 
\[ 2 \| p \vphi \|^2_2 \leq \| V_N / |p| \|_2 \| p \vphi \|_2 \, .\]
Dividing by $\| p \vphi \|_2$ and squaring, we obtain  
\begin{equation}\label{eq:pvh-bd}
\|p\vphi\|_2^2  \lesssim \sum_p \frac{|\hat V_N(p)|^2}{p^2} \lesssim \| \hat{V}_N \|_\infty^2  \| |p|^{-2} \chi_{|p| < N}  \|_1 +  \| \hat{V}_N \|_\infty \| \hat{V}_N \|_2 \| |p|^{-2} \chi_{|p| > N} \|_2 \lesssim N^{-1}.
\end{equation}
Using again \eqref{scattequ} we obtain the pointwise bound
\begin{equation}\label{maxnormvphi}
 | p^2 \vphi_p| \leq  |\hat{V}_N (p)| + \Big[ \sum_q \frac{|\hat{V}_N (p-q)|^2}{q^2} \Big]^{1/2} \| q \vphi \|_2 \lesssim N^{-1} 
 \end{equation}
where we proceeded as in (\ref{eq:pvh-bd}) to bound $\| |\hat{V}_N|^2 * |q|^{-2} \|_\infty$. This proves (\ref{eq:point}) and immediately implies the bounds for $\| \tilde \vphi \|_2 , \| \tilde \vphi \|_\infty$. To obtain the bound on $ \| \tilde \vphi \|_1$ we divide \eqref{scattequ} by $|p|^2$. Proceeding as in (\ref{eq:pvh-bd}), we obtain 
\begin{align*}
\sum_{p\neq 0} \frac{|\hat{V}_N (p)|}{|p|^2} \lesssim 1 \, ,
\end{align*}
hence we only have to bound $\||p|^{-2} (\hat{V}_N \ast \vphi) \|_{1}$. Iterating \eqref{scattequ} and using the regularizing estimate $\||p|^{-2} \hat{V}_N \ast g\|_{6p/(6+p) + \varepsilon} \leq C_{\varepsilon} \|\hat{V}_N\|_2 \|g\|_{p}$ for all $\varepsilon >0$, $p \geq 6/5$, $g \in \ell^{p}(\Lambda^*)$ and some $C_{\varepsilon} > 0$, we obtain that $\|\vphi\|_1 < \infty$. Separating high and low momenta, we obtain for $A \geq 1$ and $\varepsilon >0$,
\begin{align*}
\|\varphi\|_{1} 
	&\lesssim 1+ \| \chi_{|p| > A N} |p|^{-2} \|_{2} \|V_N \check \vphi\|_{2} + \|\chi_{|p| \leq A N}\|_{1} \| \hat{V}_N \ast \vphi \|_{\infty} \\
	&\lesssim 1 + A^{-\tfrac{1}{2}} \|\vphi\|_{1} + A \, ,
\end{align*}
where we used that $\|\check \vphi\|_{\infty} \leq \|\vphi\|_1$ and the H\"older inequality as in (\ref{maxnormvphi}) to estimate $\| \hat{V}_N \ast \vphi \|_{\infty}$. Taking $A$ sufficiently large but fixed we obtain $\|\varphi\|_{1} \lesssim 1$. 

Eq. (\ref{eq:approx_aN}) follows by noticing that, from the definition (\ref{eq:aNdef}), 
\[ \Big| 8\pi \mathfrak{a}_N - \hat{V} (0) - N \sum_p V_N (p) \tilde\vphi_p \Big| \leq N \sum_{|p| \leq N^\alpha} |\hat{V}_N (0)| |\vphi_p| \lesssim \frac{1}{N} \sum_{|p| \leq N^\alpha}\frac{1}{|p|^2} \lesssim N^{-1+\alpha} \]
where we used (\ref{eq:point}) and $\| \hat{V}_N \|_\infty \lesssim N^{-1}$. 
\end{proof}

Using the bounds in Lemma \ref{Lemmavphi} we can control the growth of the number of excitations w.r.t. the action of $\cB_2$; the proof of the next lemma can be found, for example, in \cite[Lemma 3.1]{BS}. 
\begin{lemma}\label{Gron}
For every $n \in \mathbb{N}$ and $|s| \leq 1$, we have  
	\begin{align*}
	\pm (e^{-s\B_2}\N_+ e^{s\B_2} - \N_+) 
		&\lesssim C N^{-\alpha/2} (\N_++1), \\
	e^{-s\B_2}(\N_++1)^n e^{s\B_2} 
		&\lesssim (\N_++1)^n.
	\end{align*}
\end{lemma}

In the next Proposition, we describe the action of the operator $\cB_2$, defined as in (\ref{defB}), on the Hamilton operator (\ref{eq:decompo_HN}). 
\begin{proposition}\label{propB2}
We have 
\begin{equation} \label{eq:B2} e^{-\B_2} H_N e^{\B_2} = 4 \pi \aN (N-1) + \sum_{|p|\leq N^{\alpha}} \frac{(4\pi \aN)^2}{p^2} + H_1 + \tilde H_2 + \tilde Q_2 + Q_3 + Q_4 +  \E_{\B_2},\end{equation} 
with 
\begin{align}
\tilde H_2 &= (2\hat V(0) - 8\pi \aN) \N_+, \label{tH2} \\
\tilde Q_2 &= \sum_{|p| \leq N^\alpha} 4\pi \mathfrak{a}_N \Big[ a_p^\dagger a_{-p}^\dagger \frac{a_0 a_0}{N} + {\rm h.c.} \Big] , \label{tQ2}
\end{align}
and
\[
\pm\E_{\B_2} \lesssim N^{-\alpha/2} Q_4 + \big[N^{-\alpha/2} + N^{-1+5\alpha/2} \big] (\cN_+ + 1)  + N^{-1+\alpha} \cN_+^2  + N^{-2} H_1. \]
\end{proposition}

In order to show Prop. \ref{propB2}, we define
\begin{align}\label{realscattequ}
\Gamma_2 := [H_1 + Q_4,\B_2] + Q_2 - \tilde Q'_2 
\end{align}
with 
\begin{equation} \label{eq:tQ'2} \tilde Q'_2 = \sum_p \hat W (p) a^\dagger_p a^\dagger_{-p} a_0 a_0 +  {\rm{h.c.}}, \end{equation} 
and 
\[ \hat{W} (p) = \frac{1}{2} \chi_{|p| \leq N^\alpha} \Big[ \sum_q \hat{V}_N (p-q) \vphi_q + \hat{V}_N (p) \Big] - \frac{1}{2} \sum_{|q| \leq N^\alpha} \hat{V}_N (p-q) \vphi_q. \]
We observe that 
\begin{equation} 
\begin{split} 
e^{-\B_2} &H_N e^{\B_2}  \\ &= H_0+  H_1 + Q_4 + \int_0^1 e^{-t\B_2} [H_1+Q_4,\B_2] e^{t\B_2} \dt +  e^{-\B_2} Q_2 e^{\B_2} + e^{-\B_2}(H_2 + Q_3)e^{\B_2}  
\\ &= H_0+  H_1 + Q_4 + \int_0^1 e^{-t\B_2} (-Q_2 + \tilde Q'_2 + \Gamma_2) e^{t\B_2} \dt +  e^{-\B_2} Q_2 e^{\B_2} + e^{-\B_2}(H_2 + Q_3)e^{\B_2}  \\ 
& = H_0+  H_1 + \tilde{Q}'_2 + Q_4 + \int_0^1 \int_s^1 e^{-t{\B_2}}[Q_2,{\B_2}]e^{t{\B_2}} \dt \ds+ \int_0^1 \int_0^s e^{-t \cB_2} \tilde{Q}'_2 e^{t \cB_2} \dt \ds \\ &\qquad + \int_0^1 e^{-t\B_2}  \Gamma_2 e^{t\B_2} \dt + e^{-\B_2}(H_2 + Q_3)e^{\B_2} \label{BHm2}
\end{split}
\end{equation} 
where in the last step we used 
\[ \begin{split} 
\int_0^1 e^{-t \cB_2} \tilde Q'_2 e^{t \cB_2} \dt &= \tilde Q'_2 + \int_0^1 \int_0^s e^{-t \cB_2} [\tilde{Q}'_2 , \cB_2 ] e^{t \cB_2} \dt \ds \\
e^{-\B_2} Q_2 e^{\B_2} - \int_0^1 e^{-t\B_2} Q_2 e^{t\B_2} \dt &=  \int_0^1 \int_s^1 e^{-t{\B_2}}[Q_2,{\B_2}]e^{t{\B_2}} \dt \ds \, . 
\end{split} \]
The proof of Prop. \ref{propB2} now follows by controlling the terms on the r.h.s. of \eqref{BHm2}. This is accomplished through a series of lemmas. We start by controlling the contribution arising from $H_2$. 

\begin{lemma}
\label{H2} 
On $\{\N = N\}$, we have
\begin{equation}
e^{-\B_2} H_2 e^{\B_2} =  \hat V (0) \N_+ + \E_{H_2} , 
\end{equation}
with
$$ \pm \E_{H_2} \lesssim N^{-\alpha/2} (\N_++1) + \N_+^2/N + N^{-2} H_1.$$
\end{lemma}

\begin{proof}
We have
\begin{multline}
    e^{-\B_2}H_2e^{\B_2}=N \sum_{p} \hat V_N(p) \left(a^\dagger_p a_p + \int_0^1 e^{-s\B_2} [ a^\dagger_p a_p,\B_2] e^{s\B_2} \ds\right)
    \\- e^{-\B_2} \left( \sum_{p} \hat V_N(p) a^\dagger_p a_p\N_+ + \frac {\hat V_N(0) }{2} \N_+(\N_+ -1) \right) e^{\B_2}.
\end{multline}
The term on the second line is controlled using Lemma \ref{Gron} by $N^{-1} \N_+^2$. As for the second term in the parenthesis in the first line, we use Lemma \ref{Gron} to estimate 
\begin{align*}
    \pm \sum_{p} N \hat V_N(p)  [ a^\dagger_p a_p,\B_2]   
    &= \pm \sum_{p} N \hat V_N(p) \tilde \vphi_p a^\dagger_p a^\dagger_{-p} a_0 a_0 +\text{h.c.} \\
    &\leq C \|\tilde \vphi \|_{2} (\N_0+1) (\N_++1) \leq N^{- \alpha/2} (\N_+ + 1),
\end{align*}
where we used that $\N_0\lesssim N$ and Lemma \ref{Lemmavphi}. Finally, since $V_N$ is even, we obtain 
\begin{align*}
N |\hat V_N (p) - \hat V_N(0)| \leq \|x^2V\|_{1} N^{-2} p^2 \leq C N^{-2} p^2 
\end{align*}
which gives
\begin{align*}
\pm \left(N \sum_{p} \hat V_N(p) a^\dagger_p a_p - \hat V(0) \N_+\right) \leq N^{-2} H_1 \, .
\end{align*}
\end{proof}

 The estimate of the term involving $Q_3$ is obtained analogously to \cite{BBCS4,Hainzl}. We repeat the proof for the sake of completeness. 
\begin{lemma}\label{Q3}
We have 
$$e^{-\B_2} Q_3 e^{\B_2}  =   Q_3 + \E_{Q_3} ,    $$
with
$$\pm \E_{Q_3} \lesssim N^{-\alpha/2} (Q_4 + \N_+ +1) .$$
\end{lemma}
\begin{proof}
We can rewrite 
\begin{equation}\label{43}
e^{-\B_2} Q_3 e^{\B_2} =  \sum_{q,r} \hat{V}_N(r)\left[e^{-\B_2}a_{q+r}^\dagger a_{-r}^\dagger e^{\B_2} e^{-\B_2}  a_{q}  a_0 e^{\B_2}  + \text{h.c.}  \right].
\end{equation}
 Via Duhamel's formula we have
\begin{equation}\label{Duham}
 e^{-\B_2}a_{q+r}^\dagger a_{-r}^\dagger e^{\B_2} = a^\dagger_{q+r} a^\dagger_{-r} + \int_0^1 e^{-s\B_2} [a^\dagger_{q+r} a^\dagger_{-r}, \B_2] e^{s\B_2} \ds
 \end{equation}
 and
 $$ e^{-\B_2}  a_{q}  a_0 e^{\B_2} = a_q a_0 + \int_0^1 e^{-s\B_2} (\tilde \vphi_q a^\dagger_{-q} a_0 a_0 a_0 - a_q \sum_l \tilde \vphi_l a_{-l} a_l a_0^\dagger) e^{s\B_2} \ds.$$
 The product of the first terms, combined with its hermitian conjugate, corresponds to $Q_3$ in the statement of the lemma. All other terms will be estimated in three steps. 
 %As the final bound is selfadjoint, we will w.l.o.g. omit the hermitian conjugate from now on i%n \eqref{43}.
 
 \textit{Step 1.}  Passing to $x$-space, we find 
 \begin{equation}\label{firstterm}
 \begin{split} 
\Big| \sum_{q,r}& \tilde \vphi_q  \hat{V}_N(r)  \int_0^1  \big\langle \xi , a_{q+r}^\dagger a_{-r}^\dagger  e^{-s\B_2} a^\dagger_{-q} a_0 a_0 a_0 e^{s\B_2} \xi \rangle \, \ds \Big|  \\ &=  \int_0^1  \int_{\Lambda^2} \dx  \dy   V_N(x-y) \big\langle \xi, \ca^\dagger_x \ca_y^{\dagger} e^{-s\B_2}  \ca^\dagger(\check{\tilde\vphi}_x)  a_0 a_0 a_0 e^{s\B_2} \xi \big\rangle \, \ds \\ 
&\leq  \Big( \int_{\Lambda^2} \dx  \dy   V_N(x-y) \| \ca_x \ca_y \xi \|^2 \Big)^{1/2}  \\ &\hspace{2cm} \times \Big( \int_{\Lambda^2} \dx  \dy  V_N (x-y) \int_0^1 \| \ca^\dagger(\check{\tilde\vphi}_x)  a_0 a_0 a_0 e^{s\B_2} \xi \|^2 \ds \Big)^{1/2}  \\
&\leq C N^{3/2} \| V_N \|_1^{1/2} \| \tilde\vphi \|_2 \| Q_4^{1/2} \xi \| \| (\cN_+ + 1)^{1/2} \xi \| 
\\ &\leq C N^{-\alpha/2} \langle \xi , (Q_4 + \cN_+ + 1) \xi \rangle 
\end{split} 
 \end{equation}
where we used Lemma \ref{Lemmavphi}, Lemma \ref{Gron} and the bound $\cN_0 \leq N$. 

\textit{Step 2.} Similarly, we have 
\begin{equation}\label{secterm}
\begin{split} 
 \Big|   \sum_{q,r} &\hat{V}_N(r)  \sum_l  \tilde \vphi_l  \int_0^1 \ds \, \big\langle \xi, a_{q+r}^\dagger a_{-r}^\dagger e^{-s\B_2} a_{q} a_{-l} a_l a^\dagger_0 e^{s\B_2} \xi \big\rangle  \Big| \\ &= \Big| \sum_l  \tilde \vphi_l \int_0^1 \ds \int_{\Lambda^2} \dx \dy \, V_N(x-y) \langle \xi ,   a^\dagger_x a^\dagger_y e^{-s\B_2} a_x a_{-l} a_l a^\dagger_0 e^{s\B_2} \xi \rangle 
 \Big|  \\ &\leq C \| \tilde \vphi \|_2 \|V_N \|^{1/2}_1 \| Q_4^{1/2} \xi \|  \| (\cN_++1)^{2} \xi \| \leq C N^{-\alpha/2} \langle \xi , Q_4 + \cN_+ + 1) \xi \rangle  \, .
\end{split} 
\end{equation}

\textit{Step 3.} The remaining term has the form
$$  \sum_{q,r} \hat{V}_N(r)\int_0^1e^{-s\B_2}[a_{q+r}^\dagger a_{-r}^\dagger,\B_2] e^{(s-1)\B_2}  a_{q}  a_0  e^{\B_2} \ds $$
A straightforward computation gives
$$ [a_{q+r}^\dagger a_{-r}^\dagger,\B_2] = -  \left( \tilde\vphi_r a^\dagger_{q+r} a_r + \tilde\vphi_{-r -q} a^\dagger_{-r} a_{-q-r} \right)a^\dagger_0a^\dagger_0.$$
The contributions of the two terms in the parenthesis can be handled similarly. Let us consider, for example, the expectation 
\begin{equation}\label{re1}
\begin{split} 
\Big| \sum_{q,r} \hat{V}_N(r) \tilde \vphi_r \int_0^1& \langle \xi, e^{-s\B_2}a_{q+r}^\dagger a_{r} a^\dagger_0 a^\dagger_0e^{(s-1)\B_2}  a_{q}  a_0  e^{\B_2} \xi \rangle  \,  \ds \Big|  \\
&\lesssim \frac{1}{N} \sum_{r,q} |\tilde\vphi_r| \int_0^1 \| a_{q+r} a_0 a_0 e^{(s-1)\B_2} \xi\| \| a_r e^{s \B_2} a_q a_0  e^{\B_2} \xi \| \ds \\
&\lesssim \| \tilde \vphi \|_2  \| \cN_+^{1/2} \xi \|  \int_0^1 \Big( \sum_{r,q} \| a_r e^{(s-1) \B_2} a_q a_0 e^{\B_2} \xi \|^2 \Big)^{1/2} \ds \\
 &\lesssim \| \tilde \vphi\|_2 \| \cN_+^{1/2} \xi \| \| (\cN_+ +1)^{3/2} \xi \|  \leq N^{-\alpha/2} \langle \xi , (\cN_+ + 1) \xi \rangle 
\end{split} 
\end{equation}
where we used Lemma \ref{Lemmavphi} and Lemma \ref{Gron}.
\end{proof}

Next, we recall the definition (\ref{realscattequ}) and we consider the term containing $\Gamma_2$, appearing on the r.h.s. of (\ref{BHm2}). 
\begin{lemma}\label{lemma:Gamma2}
We have
\begin{equation} \label{eq:G2} \Gamma_2 =   \sum_{r,p,q} \hat{V}_N(r) \tilde \vphi_p[a_{p+r}^\dagger a_q^\dagger a_{-p}^\dagger a_{q+r}a_0a_0 + \text{h.c.}] \end{equation} 
and 
\begin{equation} \label{eq:G2s} 
\int_0^1 e^{-t\B_2} \Gamma_2 e^{t\B_2} \dt \lesssim N^{-\alpha/2} (Q_4 + \cN_+ + 1).
\end{equation} 
\end{lemma}
\begin{proof} 
Straightforward calculations yield  
$$
[H_1,\B_2] = \sum_{p }  p^2 \tilde \vphi_p [a_p^\dagger a_{-p}^\dagger a_0 a_0  + \text{h.c.}]
$$
and
\begin{multline}\label{eq:Q4B2} 
[Q_4, \B_2] = \frac{1}{2} \sum_{p,q} \hat{V}_N(p-q) \tilde \vphi_q a_{p}^\dagger a_{-p}^\dagger a_0a_0
+ \sum_{r,p,q} \hat{V}_N(r) \tilde \vphi_p a_{p+r}^\dagger a_q^\dagger a_{-p}^\dagger a_{q+r}a_0 a_0 + \text{h.c.}
\end{multline}
Hence,
\[ \begin{split} 
[H_1 + Q_4,\B_2] + Q_2 = \; &\sum_{|p| > N^\alpha} \Big(p^2  \vphi_p + \frac{1}{2} \sum_{q} \hat{V}_N(p-q) \vphi_q + \frac{1}{2} \hat{V}_N (p)\Big) (a_p^\dagger a_{-p}^\dagger a_0 a_0+ \text{h.c.} ) \\ &- \frac{1}{2} \sum_{|p| > N^\alpha , |q| \leq N^\alpha} \hat{V}_N (p-q) \vphi_q 
(a_p^\dagger a_{-p}^\dagger a_0 a_0+ \text{h.c.} ) \\ &+ \frac{1}{2} \sum_{|p| \leq N^\alpha} \Big( \sum_{|q| > N^\alpha} \hat{V}_N (p-q) \vphi_q + \hat{V}_N (p) \Big) (a_p^\dagger a_{-p}^\dagger a_0 a_0+ \text{h.c.} )  \\ 
&+ \sum_{r,p,q} \hat{V}_N(r) \tilde \vphi_p a_{p+r}^\dagger a_q^\dagger a_{-p}^\dagger a_{q+r}a_0 a_0 \\ 
= \: &\tilde Q'_2 + \sum_{r,p,q} \hat{V}_N(r) \tilde \vphi_p a_{p+r}^\dagger a_q^\dagger a_{-p}^\dagger a_{q+r}a_0 a_0 
\end{split} \] 
where we used the scattering equation \eqref{scattequ} and (\ref{eq:tQ'2}). Comparing with (\ref{realscattequ}), we find (\ref{eq:G2}).

To prove (\ref{eq:G2s}), we write 
\[ \int_0^1 e^{-t \cB_2} \Gamma_2 e^{t \cB_2} \dt = 
\sum_{p,q,r} \hat{V}_{N}(r)\tilde\vphi_p \int_0^1  e^{-t\B_2} a_{p+r}^\dagger a_q^\dagger e^{t\B_2}e^{-t\B_2} a_{-p}^\dagger a_{q+r} a_0 a_0 e^{t\B_2} \dt  \]
and we proceed similarly as in Lemma \ref{Q3}. We omit further details. 
\end{proof}

Next, we focus on the contribution with the commutator $[Q_2, \cB_2]$ on the r.h.s. of (\ref{BHm2}). 
\begin{lemma}
\label{lemma:Q2B2} On $\{ \cN = N \}$, we have 
\begin{equation}
\begin{split} 
 \int_0^1 & \int_s^1 e^{-t\B_2}[Q_2,\B_2]e^{t\B_2} \dt \ds 
\\ 	&= - \frac {(N-1)}{2} (\hat V(0) -8\pi \aN )  -\frac { N(N-1)}{2} \sum_{|p|\leq N^\alpha} \hat{V}_N(p) \vphi_p   + \N_+  (\hat V(0) -8\pi \aN ) + \mathcal{E}_{[Q_2,\B_2]}  \label{CQ2B2},
 \end{split}
 \end{equation} 
and
\begin{equation}
\pm \mathcal{E}_{[Q_2,\B_2]}   \lesssim (N^{-\alpha/2} + N^{\alpha -1}) \cN_+ + N^{-1} (\cN_+ + 1)^2 + N^{-\alpha/2} Q_4. \label{eq:Q2B2}
\end{equation}
\end{lemma}
\begin{proof}
First, we claim that 
\begin{align}
 \int_0^1 \int_s^1 e^{-t\B_2}[Q_2,\B_2]e^{t\B_2} \dt \ds 
 	&= \frac { N(N-1)}{2} \sum_p \hat{V}_N(p) \tilde \vphi_p   
 -N \N_+ \sum_p  \hat{V}_N(p) \tilde \vphi_p +  \mathcal{E}'_{[Q_2,\B_2]}  \label{CQ2B}
\end{align}
with
\begin{multline}
 	\label{Xi}
 \mathcal{E}'_{[Q_2,\B_2]}  = - 2N \sum_p \hat{V}_N (p) \tilde\vphi_p \int_0^1 \int_s^1 \left[ e^{-t \cB_2} \cN_+ e^{t \cB_2} - \cN_+ \right] \dt \ds \\ + \sum_p  \hat{V}_N(p) \tilde\vphi_p   \int_0^1 \int_s^1 e^{-t\B_2} \N_+( \N_+ + 1 )e^{t\B_2}\dt \ds \\
  +2\sum_p \hat{V}_N(p) \tilde \vphi_p \int_0^1 \int_s^1 e^{-t\B_2}\N_0(\N_0 -1) a^\dagger_p a_pe^{t\B_2} \dt \ds  \\
 -  \sum_{p,q}  \hat{V}_N(p) \tilde \vphi_q \int_0^1 \int_s^1 e^{-t\B_2}  a^\dagger_{p} a^\dagger_{-p} a_{-q} a_{q} (2 \N_0+1 )  e^{t\B_2} \dt \ds \, .
\end{multline}
To prove \eqref{Xi}, we  calculate
\begin{align}\nonumber
[Q_2,\B_2] &= \frac{1}{4} \sum_{p,q} \hat{V}_N(p) \tilde\vphi_q [a_p^\dagger a_{-p}^\dagger a_0a_0 + a_{-p}a_pa_0^\dagger a_0^\dagger , a_q^\dagger a_{-q}^\dagger a_0a_0 - a_{-q}a_q a_0^\dagger a_0^\dagger]
\\ \label{Q2B}
%&= \kappa \frac{\mathcal{N}_0}{4} \sum_{p,q} \hat{V}_N(p) \vphi_p \left([a_{-p} a_p,a_q^\dagger a_{-q}^\dagger] + [a_{-q} a_q,a_p^\dagger a_{-p}^\dagger]\right)
%\\
&=  \frac{1}{4} \sum_{p,q} \hat{V}_N(p) \tilde\vphi_q \left( [a_{-p}a_pa_0^\dagger a_0^\dagger , a_q^\dagger a_{-q}^\dagger a_0a_0] - [a_p^\dagger a_{-p}^\dagger a_0a_0,a_{-q}a_q a_0^\dagger a_0^\dagger] \right).
\end{align}
The two terms in the bracket are hermitian conjugates. Hence, it suffices to compute the second one
\begin{align*}
  - [a_p^\dagger a_{-p}^\dagger a_0a_0,a_{-q}a_q a_0^\dagger a_0^\dagger] &= [a_{-q}a_q, a_p^\dagger a_{-p}^\dagger] a_0^\dagger a_0^\dagger a_0a_0 - a_p^\dagger a_{-p}^\dagger a_{-q} a_q [a_0a_0,a_0^\dagger a_0^\dagger],
 \end{align*}
where
 \begin{equation} \label{eq:comm1}[a_{-p}a_p,a_q^\dagger a_{-q}^\dagger ] = (\delta_{p,q} + \delta_{p,-q})( 1 + a^\dagger_p a_p + a^\dagger_{-p} a_{-p} ),\end{equation} 
 and
 \begin{equation}\label{eq:comm2} \begin{split} a_0^\dagger a_0^\dagger a_0 a_0 &= \N_0(\N_0-1)= N(N -1) - 2 N \N_+ + \N_+(\N_+ +1), \\  [a_0a_0,a^\dagger_0a^\dagger_0] &= 2(2 \N_0 + 1). \end{split} \end{equation} 
 Inserting these identities on the r.h.s. of (\ref{Q2B}), conjugating with $e^{t\cB_2}$ and integrating over $t,s$ 
 %and further expanding 
%\[ \begin{split} \sum_p \hat{V}_N &(p) \tilde\vphi_p  \int_0^1 \int_s^1 e^{-t \cB_2} 2 N  
%\cN_+ e^{t \cB_2} \dt \ds \\ = \; &N \cN_+ \sum_p \hat{V}_N (p) \tilde\vphi_p + \sum_p 
%\hat{V}_N (p) \tilde\vphi_p \int_0^1 \int_s^1 \int_0^t e^{-t_2 \cB_2} \big[ 2 N \cN_+ , \cB_2 %\big] e^{t_2 \cB_2} \dt_2 \dt \ds \end{split} \] 
 we obtain (\ref{Xi}).
 
With the definition (\ref{eq:aNdef}), we write 
\[ \begin{split} \frac{N(N-1)}{2} \sum_p \hat{V}_N (p) \tilde{\vphi}_p &= \frac{(N-1)}{2} (\hat{V} (0) - 8 \pi \aN) - \frac{N(N-1)}{2} \sum_{|p| \leq N^\alpha} \hat{V}_N (p) \vphi_p , \end{split} \]
and we use \eqref{eq:approx_aN} to estimate 
\[  \pm \Big[ - N \cN_+ \sum_p \hat{V}_N (p) \tilde\vphi_p + (8\pi \aN - \hat{V} (0)) \mathcal{N}_+ \Big] \lesssim N^{\alpha-1} \cN_+ \, . \]
Thus, Lemma \ref{lemma:Q2B2} follows from (\ref{CQ2B}), if we can prove that 
$\mathcal{E}'_{[Q_2,\B_2]}$ satisfies the estimate (\ref{eq:Q2B2}).  

Using the bound
\[ \Big| \sum_p \hat{V}_N (p) \tilde\vphi_p \Big| \lesssim N^{-1} \]
and Lemma \ref{Gron}, we can bound the first term on the r.h.s. of (\ref{Xi}) by 
\[ \pm 2N \sum_p \hat{V}_N (p) \tilde\vphi_p \int_0^1 \int_s^1 \left[ e^{-t \cB_2} \cN_+ e^{t \cB_2} - \cN_+ \right] \dt \ds \lesssim N^{-\alpha/2} (\cN_+ + 1) \, .\]
Also the second term on the r.h.s. of (\ref{Xi}) can be bounded with Lemma \ref{Gron}; we find 
\[ \pm \sum_p \hat{V}_N (p) \tilde\vphi_p \int_0^1 \int_s^1 e^{-t \cB_2} \cN_+ ( \cN_+ + 1) e^{t \cB_2} \dt \ds \lesssim N^{-1} (\cN_+ + 1)^2\,. \]
As for the third term, we use $\| \hat{V}_N \|_\infty \lesssim N^{-1}$, $\| \tilde \vphi \|_\infty \lesssim N^{-1-2\alpha}$ together with $\cN_0 (\cN_0 -1) \leq N^2$ and again Lemma \ref{Gron} to conclude that 
\[ \pm 2 \sum_p \hat{V}_N (p) \tilde\vphi_p \int_0^1 \int_s^1 e^{-t \cB_2} \cN_0 ( \cN_0 + 1) e^{t \cB_2} \dt \ds \lesssim N^{-2\alpha} \cN_+ \,.  \]
To control the last term on the r.h.s. of (\ref{Xi}), we write 
\begin{multline}\label{eq:last} 
 \sum_{p,q}  \hat{V}_N(p) \tilde \vphi_q \int_0^1 \int_s^1 e^{-t\B_2}  a^\dagger_{p} a^\dagger_{-p} a_{-q} a_{q} (2 \N_0+1)  e^{t\B_2} \dt \ds \\ =
\sum_{p}  \hat{V}_N(p) \int_0^1 \int_s^1 e^{-t\B_2}  a^\dagger_{p} a^\dagger_{-p} e^{t\B_2} e^{-t\B_2} \Phi (2 \N_0+1) e^{t\B_2} \dt \ds
\end{multline}
where we defined $\Phi = \sum_q \tilde \vphi_q a_{-q} a_q$ so that, by Lemma \ref{Lemmavphi}, 
\begin{equation}\label{eq:Phibd} \| \Phi \xi \| \lesssim \| \tilde\vphi \|_2 \| \cN_+ \xi \| \lesssim N^{-1-\alpha/2} \| \cN_+ \xi \| \,.\end{equation} 
 Next, we expand  
\begin{equation}\label{adad} e^{- t\B_2}  a^\dagger_{p} a^\dagger_{-p} e^{t\B_2}  =  a^\dagger_{p} a^\dagger_{-p} + \int_0^t e^{-\tau \B_2} [a^\dagger_{p} a^\dagger_{-p} ,\B_2] e^{\tau \B_2} d \tau .
\end{equation}  
Inserting this identity into (\ref{eq:last}), we obtain two contributions. The first contribution can be controlled passing to position space. We find 
\[ \begin{split} 
\pm \sum_{p}  \hat{V}_N(p) &\int_0^1 \int_s^1 a^\dagger_{p} a^\dagger_{-p} e^{-t\B_2} \Phi (2 \N_0+1) e^{t\B_2} \dt \ds \\  &=  \pm \int_0^1 \int_s^1 \int_{\Lambda^2} \dx \dy  \kappa V_N(x-y) \ca_x^\dagger \ca_y^\dagger e^{-t\B} \Phi (2 \N_0 +1) e^{t\B} \dt \ds + \text{h.c.} \\
	&\lesssim \delta Q_4 + \delta^{-1} N^2 \|\varphi\|_{L^2}^2 \|V_N\|_{1}   (\N_++1)^2 \lesssim N^{-\alpha/2} Q_4 + N^{-1-\alpha/2}  (\N_++1)^2 .\end{split} \]
On the other hand, the contribution arising from the second term on the r.h.s. of (\ref{adad})  can be controlled by
\begin{multline}
\pm \sum_{p}  \hat{V}_N(p) \tilde \vphi_p \int_0^1 \int_s^1\int_0^t e^{-\tau \B_2}   a^\dagger_0 a^\dagger_0 (2 a^\dagger_p a_p + 1 )e^{\tau \B_2}   e^{-t\B} \Phi (2 \N_0 +1) e^{t\B} \dt \ds + \text{h.c.} \\ \lesssim N^{-1 -\alpha/2} (\N_++1)^2 .
\end{multline} 
This concludes the proof of (\ref{CQ2B2}), (\ref{eq:Q2B2}). 
\end{proof}

Finally, we control the contribution with the commutator $[ \tilde{Q}_2 , \cB_2 ]$ in (\ref{BHm2}).
\begin{lemma} \label{F2}
We have 
\begin{equation}\label{inttQ2B2}
 \int_0^1 \int_0^s e^{-t\B_2}[\tilde Q_2,\B_2] e^{t\B_2} \dt \ds=  - \frac {N(N-1)} 2   \sum_{|p|> N^\alpha, |q|\leq N^\alpha} \hat{V}_N(p-q) \vphi_p \vphi_q  + \E_{[\tilde Q_2,\B_2]},
 \end{equation}
 with $$ \pm \E_{[\tilde Q_2 , \B_2]} \lesssim N^{-\alpha/2} Q_4 + N^{-1 + \alpha} (\N_++1)^2 \,.$$
\end{lemma}
\begin{proof}
     Using that $ \chi_{|p| \leq N^\alpha} \tilde \varphi_p = 0$, similar computations as in the proof of Lemma \ref{CQ2B} yield
    \begin{equation}\begin{split} 
        [\tilde Q_2 &,\B_2] \\ = \; &\frac{1}{4} \sum_{|p|>N^\alpha, |q|\leq N^\alpha,r} \hat V_N(p-q) \vphi_q \tilde \vphi_r \\ &\times  (a^\dagger_p a^\dagger_{-p} a_r a_{-r}[a_0 a_0,a^\dagger_0 a^\dagger_0] +[a^\dagger_p a^\dagger_{-p},a_r a_{-r}]a^\dagger_0 a^\dagger_0 a_0 a_0)+\text{h.c.}
        \\ = \; &\sum_{|p|>N^\alpha, |q| \leq N^{\alpha}} \frac{1}{4}\hat V_N(p-q) \vphi_q \\ & \times \big( (2 a^\dagger_p a^\dagger_{-p} \Phi (2\N_0+1) + \text{h.c.}) 
   - 4 \tilde \vphi_p (a^\dagger_p a_p +a^\dagger_{-p} a_{-p}) \N_0 (\N_0-1)  -4\tilde\vphi_p \N_0 (\N_0-1)\big) \label{tQ2B2}
\end{split}
\end{equation} 
with the notation $\Phi = \sum_r \tilde\vphi_r a_r a_{-r}$. To bound the contribution arising from the first term in the parenthesis, we decompose 
\begin{equation}\label{eq:dec-1} \begin{split} \frac{1}{2} &\sum_{|p|>N^\alpha, |q|\leq N^\alpha} \hat V_N(p-q) \vphi_q a^\dagger_p a^\dagger_{-p} \Phi (2\N_0+1)\\  = \; &\frac{1}{2}   \sum_{p,|q|\leq N^\alpha} \hat V_N(p-q) \vphi_q a^\dagger_p a^\dagger_{-p} \Phi (2\N_0+1) - \frac{1}{2}   \sum_{|p|, |q|\leq N^\alpha} \hat V_N(p-q) \vphi_q a^\dagger_p a^\dagger_{-p} \Phi (2\N_0+1)\,. \end{split} \end{equation} 
The first term can be controlled switching to position space. With the notation $\check{\vphi}^<$ for the Fourier series of $\chi_{|q| \leq N^\alpha} \vphi_q$, we find 
\[ \begin{split} 
\pm \frac{1}{2}   &\sum_{p,|q|\leq N^\alpha} \hat V_N(p-q) \vphi_q a^\dagger_p a^\dagger_{-p} \Phi (2\N_0+1) \\ &= \pm \int_{\Lambda^2} \dx  \dy  V_N (x-y) \check \vphi^< (x-y) a^\dagger_x  a^\dagger_y \Phi (2\N_0+1) + \text{h.c.} \nonumber \\
	&\lesssim  \delta Q_4 + \delta^{-1} N^{-\alpha} \|V_N^{1/2} \check\vphi^< \|_2^2 (\cN_++1)^2 
\end{split} \]
where we used the bound (\ref{eq:Phibd}) for $\Phi$ and $\cN_0 \leq N$. With 
\begin{equation}\label{eq:Vphph} \| V_N^{1/2} \check\vphi^< \|_2^2 = \sum_{|p|, |q| \leq N^\alpha} \hat{V}_N (p-q) \vphi_p \vphi_q \lesssim \frac{1}{N^3} \Big[ \sum_{|p| \leq N^\alpha} \frac{1}{|p|^2} \Big]^2 \lesssim N^{2\alpha - 3} \end{equation} 
and choosing $\delta = N^{-\alpha/2}$, we conclude (since $\alpha < 1$) that 
\[ \pm \frac{1}{2}  \sum_{p,|q|\leq N^\alpha} \hat V_N(p-q) \vphi_q a^\dagger_p a^\dagger_{-p} \Phi (2\N_0+1) \lesssim N^{-\alpha/2} Q_4 + N^{-1-\alpha/2} (\cN_+ + 1)^2 \,.\]
As for the second term on the r.h.s. of (\ref{eq:dec-1}), we estimate 
\[ \pm \frac{1}{2}   \sum_{|p|, |q|\leq N^\alpha} \hat V_N(p-q) \vphi_q a^\dagger_p a^\dagger_{-p} \Phi (2\N_0+1) + \text{h.c.} \lesssim N^{-\alpha/2} \| \chi_{|p|\leq N^\alpha} \hat V_N \ast \vphi^{<} \|_{2} \| (\cN_+ + 1)^2 \]
where, again, we used (\ref{eq:Phibd}) and $\cN_0 \leq N$. With 
\[  \| \chi_{|p|\leq N^\alpha} \hat V_N \ast \vphi^{<}\|_{2} \leq \|\chi_{|p|\leq N^\alpha}\|_{L^2} \|V_N \check\vphi^{<}\|_{1} \lesssim N^{3\alpha/2}\| V_N^{1/2}\|_{2} \| V_N^{1/2} \check \vphi^{<}\|_{2} \lesssim N^{-2+5\alpha/2} \]
we conclude that 
\[  \pm \frac{1}{2}   \sum_{|p|, |q|\leq N^\alpha} \hat V_N(p-q) \vphi_q a^\dagger_p a^\dagger_{-p} \Phi (2\N_0+1) + \text{h.c.} \lesssim N^{-2-2\alpha} (\cN_+ + 1)^2 \,.\]
The contribution arising from the second term in the parenthesis on the r.h.s. of (\ref{tQ2B2}) can be bounded by 
$$\pm \sum_{|p|>N^\alpha, |q| \leq N^{\alpha}} \hat V_N(p-q) \vphi_q \tilde \vphi_p a^\dagger_p a_p \N_0 (\N_0-1)  \lesssim N^{-1- \alpha} \N_+,$$
using that $ \|\tilde\vphi  (\hat V_N \ast \vphi^{<})\|_{L^\infty} \lesssim  \|\tilde\vphi \|_{\infty} \| V_N \check \vphi^{<}\|_1 \lesssim N^{-3-\alpha}$. As for the contribution arising from the last term on the r.h.s. of (\ref{tQ2B2}), we write $\N_0(\N_0-1)=  N(N -1) - 2 N \N_+ + \N_+ (\N_+ +1)$. The contribution proportional to $N(N-1)$ produces the main term on the r.h.s. of  (\ref{inttQ2B2}). The other contributions can be bounded, noticing that 
\begin{align*}
\Big| \sum_{|p|> N^\alpha, |q|\leq N^\alpha}  \hat{V}_N(p-q) \vphi_p \vphi_q \Big| \leq \|\check {\tilde \vphi} V_N \check \vphi^{<}\|_{1} \leq \|\check {\tilde \vphi}\|_{\infty} \|V_N \check\vphi^< \|_1 \lesssim N^{-2+\alpha},
\end{align*}
where we used $\| \check{\tilde\vphi} \|_\infty \lesssim \| \tilde \vphi \|_1 \lesssim 1$, by Lemma \ref{Lemmavphi}. 
\end{proof}

We can now finish the proof of Prop. \ref{propB2}.

\begin{proof}[Proof of Prop. \ref{propB2}] 
Combining (\ref{BHm2}) with the bounds proven in Lemma \ref{H2}, Lemma \ref{Q3}, Lemma \ref{lemma:Gamma2}, Lemma \ref{lemma:Q2B2} and Lemma \ref{F2}, 
we conclude that 
\begin{equation}\label{eq:EHE2} e^{-\cB_2} H_N e^{\cB_2} = 4 \pi \mathfrak{a}_N (N-1) + A_\alpha + H_1 + \tilde{H}_2 +  \tilde{Q}'_2 + Q_3 + Q_4 + \mathcal{E} \end{equation} 
where
\begin{equation}\label{eq:eps-err} \pm \mathcal{E} \lesssim N^{-\alpha/2} Q_4 + \big [ N^{-\alpha/2} + N^{\alpha-1} \big] (\cN_+ + 1) + N^{-1+\alpha} \cN_+^2  + N^{-2} H_1 \end{equation} 
and where we defined 
\[ \begin{split} A_\alpha &= - \frac{N(N-1)}{2}  \Big[ \sum_{|p| \leq N^\alpha} \hat{V}_N (p) \vphi_p + \sum_{|p| > N^\alpha, |q| \leq N^\alpha} \hat{V}_N (p-q) \vphi_p \vphi_q  \Big] \\ &= - \frac{N(N-1)}{2}  \Big[ \sum_{|p| \leq N^\alpha} (\hat{V}_N (p) + \hat{V}_N * \vphi) \vphi_p - \sum_{|p|, |q| \leq N^\alpha} \hat{V}_N (p-q) \vphi_p \vphi_q \Big].
% \\ &= N (N-1) \Big[ \sum_{|p| \leq N^\alpha} p^2 |\vphi_p|^2 + \frac{1}{2} \sum_{|p|, |q| \leq %N^\alpha} \hat{V}_N (p-q) \vphi_p \vphi_q \Big] 
\end{split} \]
The second term in the parenthesis can be estimated as in (\ref{eq:Vphph}). Setting, in position space, $f = 1 + \check{\vphi}$ we find
\begin{equation}\label{eq:Aalpha} A_\alpha = - \frac{N(N-1)}{2}  \sum_{|p| \leq N^\alpha} (\hat{V}_N  * \hat{f}) (p) \vphi_p  + \mathcal{O} (N^{2\alpha-1}) .\end{equation} 
From (\ref{eq:aNdef}), we have $\widehat{V_N f} (0) = 8\pi \mathfrak{a}_N$. Hence
\begin{equation}\label{eq:diff-sc} | (\hat{V}_N * \hat{f}) (p) - 8 \pi \mathfrak{a}_N / N | \leq  \int_\Lambda V_N (x) f (x) |e^{-ip \cdot x} -1| \dx  \leq C |p| / N^2 \,.\end{equation}
Moreover, from the scattering equation (\ref{scattequ}), we find 
\[ \vphi_p = - \frac{1}{2p^2} (\hat{V}_N * \hat{f}) (p) \] 
which implies, by (\ref{eq:diff-sc}),  
\[ \Big| \vphi_p + \frac{4\pi \mathfrak{a}_N}{N p^2} \Big| \leq \frac{1}{2p^2} \big| (\hat{V}_N * \hat{f}) (p) - 8 \pi \mathfrak{a}_N \big| \leq \frac{C}{|p|N^2}\,. \]
Inserting in (\ref{eq:Aalpha}), we obtain 
\[ A_\alpha = \sum_{|p| \leq N^\alpha} \frac{(4\pi \mathfrak{a}_N)^2}{p^2} + \mathcal{O} (N^{2\alpha-1}) \,.\]
To conclude the proof of Prop. \ref{propB2}, we still have to compare the operator $\tilde Q'_2$ appearing on the r.h.s. of (\ref{eq:EHE2}) with the operator $\tilde Q_2$ defined in (\ref{tQ2}). From (\ref{eq:tQ'2}), we can write, using again the notation $f = 1+ \check{\vphi}$, 
\begin{equation} \label{eq:Q2Q2} \tilde Q'_2 - \tilde Q_2 = \frac{1}{2} \sum_{|p| \leq N^\alpha} \Big[ (\hat{V}_N * \hat{f})(p) - \frac{8\pi \mathfrak{a}_N}{N} \Big]  \, a_p^\dagger a_{-p}^\dagger a_0 a_0 - \frac{1}{2} \sum_{p, |q| \leq N^\alpha} \hat{V}_N (p-q) \vphi_q a_p^\dagger a_{-p}^\dagger a_0 a_0 + \text{h.c.} \end{equation} 
The first term on the r.h.s. of (\ref{eq:Q2Q2}) can be bounded with (\ref{eq:diff-sc}) by 
\[ \begin{split} \pm \sum_{|p| \leq N^\alpha} \Big[ (\hat{V}_N * \hat{f})(p) &- \frac{8\pi \mathfrak{a}_N}{N} \Big]  \, \big( a_p^\dagger a_{-p}^\dagger a_0 a_0 +\text{h.c.}\big) \\ &\lesssim \frac{1}{N} \| |p| \chi_{|p| \leq N^\alpha} \|_2 (\cN_+ + 1) \lesssim N^{-1+5\alpha/2} (\cN_+ + 1)\,. \end{split} \]
As for the second term on the r.h.s. of (\ref{eq:Q2Q2}), we set $\vphi^<_p = \vphi_p  \chi_{|p| \leq N^\alpha}$ and we estimate, switching to position space,  
\[ \begin{split}  \pm \sum_{p} &(\hat{V}_N * \vphi^< )(p) \, a_p^\dagger a_{-p}^\dagger a_0 a_0 +  \text{h.c.} \\ &= \pm \int_{\Lambda^2} \dx  \dy  \, V_N (x-y) \check{\vphi}^< (x-y) a_x^\dagger a_y^\dagger a_0 a_0 + \text{h.c.} \\ &\lesssim \delta Q_4 + \delta^{-1} N^2 \| V_N \|_1 \| \check{\vphi}^< \|_\infty^2 \lesssim \delta Q_4 + \delta^{-1} N^{-1+2\alpha}  \leq N^{-\alpha/2} Q_4 + N^{-1+5\alpha/2} \end{split} \]
since $\| \check{\vphi}^< \|_\infty \lesssim \| \vphi^< \|_1 \lesssim N^\alpha$, from Lemma \ref{Lemmavphi} (in the last step, we chose $\delta =  N^{-\alpha/2}$). The last two estimates show that the difference $\tilde Q'_2 - \tilde Q_2$ can be added to the error (\ref{eq:eps-err}) and conclude therefore the proof of Prop. \ref{propB2}. 
\end{proof}

\section{Cubic renormalization}

While conjugation with $e^{\B_2}$ allowed us to renormalize the quadratic part of the Hamiltonian $H_N$, regularizing the off-diagonal term $Q_2$, it did not significantly change the cubic 
operator $Q_3$. To renormalize $Q_3$, we proceed with a second conjugation, with a
unitary operator $e^{\B_3}$, where 
\begin{equation}
\label{defB3}
 \B_3 = \sum_{p,q} \tilde \vphi_p \chi_{|q| \leq N^\alpha} \, a^\dagger_{p+q} a^\dagger_{-p} a_q a_0  - \text{h.c.},
\end{equation}
with the same $0 < \alpha <1$ used in the definition (\ref{eq:tvphi}) of $\tilde\vphi$. Similarly as we did in Lemma \ref{Gron} for the action of $\B_2$, it is important to notice that conjugation with $e^{\B_3}$ does not substantially change the number of excitations. 
\begin{lemma}
	\label{lemma:gronwall_B3}
For all $s \in [-1;1]$ and all $k \in \mathbb{N}$, we have  
\begin{align}\label{eq:growth1} \pm \Big[ e^{-s\B_3} \cN_+ e^{s \B_3} - \cN_+ \Big] &\lesssim N^{-\alpha/2} (\cN_+ + 1) \, , \\ 
\label{eq:growth2} e^{-s\B_3} (\N_++1)^k e^{s\B_3} &\lesssim (\N_++1)^k \, . \end{align} 
\end{lemma}
\begin{proof}
We proceed similarly as in the proof of \cite[Prop. 5.1]{BBCS4}. For $\xi \in \cF$, we set $f(s) = \langle \xi, e^{-s\B_3} (\N_++1) e^{s\B_3} \xi \rangle$. For $s \in (0;1)$, we find 
\begin{equation}\label{eq:growth3}
\begin{split} 
f'(s)  &= \langle \xi, e^{-s\B_3}   [(\N_++1),\B_3] e^{s\B_3} \xi \rangle \\  &=  \sum_{p,q} \tilde \vphi_p  \chi_{|q| \leq N^\alpha} \, \langle \xi, e^{-s\B_3}  a^\dagger_{p+q} a^\dagger_{-p} a_q a_0 e^{s\cB_3} \xi \rangle + \text{h.c.} \\
	&\lesssim \delta  \sum_{p,q} \langle \xi, e^{-s\B_3} a^\dagger_{p+q} a^\dagger_{-p}  a_{-p} a_{p+q} e^{s\B_3} \xi \rangle + \delta^{-1}  \sum_{p,q} |\tilde \vphi_p|^2 \, \langle \xi, e^{-s\B_3}  a^\dagger_q  a_q a^\dagger_0 a_0 e^{s\B_3} \xi \rangle  \\
	&\lesssim \delta \langle e^{-s\B_3} \N_+^2  e^{s\B_3} \xi \rangle + C \delta^{-1} N  \| \tilde \vphi \|_2^2 \langle \xi , e^{-s\B_3} \N_+  e^{s \B_3} \xi \rangle \lesssim N^{-\alpha/2} f(s) \end{split}\end{equation} 
where we put $\delta = N^{-1-\alpha/2}$ and we used that $\cN_+, \cN_0 \leq N$, $\| \vphi \|_2 \lesssim N^{-1-\alpha/2}$, by Lemma~\ref{Lemmavphi}. With Gronwall's lemma \cite[Theorem 1.2.2]{Pachpatte}, we obtain $f(s) \lesssim \langle \xi, (\cN_+ + 1) \xi \rangle$, for all $s \in [-1;1]$, proving (\ref{eq:growth2}). Inserting this estimate on the r.h.s. of (\ref{eq:growth3}) and integrating over $s$, we obtain (\ref{eq:growth1}). For $k > 1$, the bound (\ref{eq:growth2}) can be shown similarly.
\end{proof}

The operator $\B_3$ is chosen (similarly as we did with $\B_2$ in Section \ref{sec:quadra}) so that the commutator $[H_1 + Q_4 , \B_3]$, arising from conjugation with $e^{\B_3}$, cancels the main part of $Q_3$. The goal of this section is to use this cancellation to prove the following proposition. 
\begin{proposition}\label{propB3}
We have
\begin{equation}\label{eq:cubic-ham}
\begin{split} 
e^{-\B_3} &e^{-\B_2} H_N  e^{\B_2} e^{\B_3}
	\\ &= 4 \pi \aN (N-1) + \frac{1}{4}\sum_{|p|\leq N^{\alpha}} \frac{(8\pi \aN)^2}{p^2} \\
	&\quad +  \sum_{p} (p^2 + 8\pi \aN  \chi_{|p| \leq N^\alpha}) a^\dagger_p a_p + \frac{1}{2}\sum_{|p|\leq N^\alpha} 8 \pi \aN [a^\dagger_p a^\dagger_{-p} \frac{a_0 a_0 }N+\text{h.c.}]  + Q_4 + \mathcal{E}_{\B_3},
\end{split} 
\end{equation}
with 
\begin{equation} \label{eq:EB3}\begin{split} \pm \mathcal{E}_{\B_3} \lesssim \; & N^{-3\alpha/2} H_1 + N^{-\alpha/2} Q_4 + N^{-\alpha/2} (\cN_+ + 1) \\ &+ N^{(3\alpha-1)/2} (\cN_+ + 1)^{3/2} + N^{-1+5\alpha/2} (\cN_+ + 1)^2.  \end{split} \end{equation} 
\end{proposition}

To prove Prop. \ref{propB3} we define
\begin{equation}\label{eq:Xi3} \Gamma_3 := \big[ H_1 + Q_4 , \cB_3 \big] + Q_3\,. \end{equation} 
Starting from (\ref{eq:B2}), we compute
\begin{align*}\nonumber
e^{-\B_3} e^{-\B_2} &H_N e^{\B_2} e^{\B_3} -  4 \pi \aN (N-1) - \sum_{|p|\leq N^{\alpha}} \frac{(4\pi \aN)^2}{p^2} -  e^{-\B_3}(\tilde H_2 + \tilde Q_2 + \mathcal E_{\B_2})e^{\B_3} \\
	&=  H_1 + Q_4 + \int_0^1 e^{-t\B_3} [H_1+Q_4,\B_3] e^{t\B_2} \dt +  e^{-\B_3} Q_3 e^{\B_3} \nonumber
\\ \nonumber
& = H_1 + Q_4 + \int_0^1 e^{-t\B_3} (-Q_3  + \Gamma_3) e^{t\B_3} \dt +  e^{-\B_3} Q_3 e^{\B_3}  
\end{align*}
which leads to 
\begin{equation}\label{BHm3} 
\begin{split} 
e^{-\B_3} e^{-\B_2} H_N e^{\B_2} e^{\B_3} = \; & 4 \pi \aN (N-1) + \sum_{|p|\leq N^{\alpha}} \frac{(4\pi \aN)^2}{p^2} +  H_1 + Q_4  + e^{-\B_3}(\tilde H_2 + \tilde Q_2 + \mathcal E_{\B_2} )e^{\B_3} \\ &+ \int_0^1 e^{-t\B_3}  \Gamma_3 e^{t\B_3} \dt + \int_0^1 \int_s^1 e^{-t{\B_3}}[Q_3,{\B_3}]e^{t{\B_3}} \dt \ds \, .\end{split} \end{equation} 
To show Prop. \ref{propB3}, we are going to control all terms on the r.h.s. of (\ref{BHm3}). We start by computing and estimating the commutator in (\ref{eq:Xi3}), defining the error term $\Gamma_3$. 
\begin{lemma}\label{H1B3}
We have 
\begin{align}
[H_1,\B_3] &=  \sum_{p,q,r}  \hat V_N(p-r) (\delta_{0,r}+\vphi_r)  \chi_{|p| > N^\alpha}  \chi_{|q| \leq N^\alpha} a^\dagger_{p+q} a^\dagger_{-p} a_q a_0 + \text{h.c.} + \mathcal{E}_{[H_1,\B_3]}, \label{eq:H1B3} \\
[Q_4,\B_3] &= \sum_{r,p,q} \hat V_N(p-r)\tilde \vphi_r  \chi_{|q| \leq N^\alpha} (a^\dagger_{p+q} a^\dagger_{-p} a_q a_0 + \text{h.c.})+ \mathcal{E}_{[Q_4,\B_3]}, \label{eq:Q4B3}
\end{align}
where  
\begin{equation} \label{eq:est_E3}
\begin{split}   
\pm\mathcal{E}_{[H_1,\B_3]} &\lesssim N^{-3\alpha/2} H_1 + N^{-1+5\alpha/2} (\N_++1)^2,  \\
\pm \mathcal{E}_{[Q_4,\B_3]} &\lesssim N^{-\alpha/2} Q_4 + N^{-1+5\alpha/2} (\N_++1)^2. 
\end{split} \end{equation}
\end{lemma}

\begin{proof}
A simple computation shows that 
\begin{align*}
 [H_1,\B_3] &= \sum_{p,q} \left[(p+q)^2 + p^2 - q^2 \right] \tilde \vphi_p \chi_{|q| \leq N^\alpha} a^\dagger_{p+q} a^\dagger_{-p} a_q a_0 + \text{h.c.}, \\
 	&= 2 \sum_{p,q} p^2 \tilde \vphi_p \chi_{|q| \leq N^\alpha} a^\dagger_{p+q} a^\dagger_{-p} a_q a_0 +\text{h.c.}+ \mathcal{E}_{[H_1,\B_3]},
\end{align*}
with
\begin{align*}
\mathcal{E}_{[H_1,\B_3]}& = 2 \sum_{p,q}   p\cdot q \tilde \vphi_p  \chi_{|q| \leq N^\alpha} a^\dagger_{p+q} a^\dagger_{-p} a_q a_0 +\text{h.c.}
\end{align*}
Using the scattering equation \eqref{scattequ} yields \eqref{eq:H1B3}. We now estimate $\mathcal{E}_{[H_1,\B_3]}$. Using $|q| \leq N^\alpha$, we find, for any $\delta>0$,
\begin{align*}
\pm \mathcal{E}_{[H_1,\B_3]}
	&= \pm 2 \sum_{p,q}  p\cdot q \tilde \vphi_p  \chi_{|q| \leq N^\alpha} a^\dagger_{p+q} a^\dagger_{-p} (\N_++1)^{-1/2} (\N_++1)^{1/2}  a_q a_0  + \text{h.c.} \\ 
	&\lesssim 
\delta \sum_{p,q} p^2 a^\dagger_{p+q} a^\dagger_{-p} (\N_++1)^{-1} a_{-p} a_{p+q} + \delta^{-1} \sum_{p,|q| \leq N^\alpha} q^2  |\tilde \vphi_p |^2 \, a^\dagger_{q} (\N_++1) a_{q}(a^\dagger_{0} a_{0}) \\
	&\lesssim \delta H_1 + \delta^{-1} N^{1+2\alpha} \|\tilde \vphi\|_2^2  (\N_++1)^2. \end{align*}
Choosing $\delta = N^{-3\alpha/2}$, we conclude that
\[ \pm \mathcal{E}_{[H_1,\B_3]} \lesssim N^{-3\alpha/2} H_1 + N^{-1+5\alpha/2} (\N_++1)^2 \, . \]

Let us now turn to (\ref{eq:Q4B3}). Recalling \eqref{eq:Qs} and \eqref{defB3}, we find
 \begin{align*}
[Q_4,\B_3] &= \frac{1}{2} \sum_{r,p,q} \sum_{m,n} \hat{V}_N (r)  \tilde \vphi_m \chi_{|n| \leq N^\alpha} \big[ a_{p+r}^\dagger a_q^\dagger a_p a_{q+r},  a_{m+n}^\dagger a_{-m}^\dagger a_n a_0 \big] + \text{h.c.} \\ &= 
 \frac{1}{2} \sum_{r,p,q} \sum_{m,n} \hat{V}_N (r)  \tilde \vphi_m \chi_{|n| \leq N^\alpha} \\ &\hspace{1.5cm} \times \Big\{ a_{p+r}^\dagger a_q^\dagger \big[ a_p a_{q+r} ,  a_{m+n}^\dagger a_{-m}^\dagger \big] a_n a_0 + a_{m+n}^\dagger a_{-m}^\dagger
\big[ a_{p+r}^\dagger a_q^\dagger ,  a_n a_0 \big] a_p a_{q+r} \Big\}\, .
\end{align*} 
Using (\ref{eq:CCR}) and rearranging all terms in normal order, we arrive at 
 \begin{align*}
 [Q_4,\B_3] &= \sum_{r,p,q} 
\hat V_N(p-r)\tilde \vphi_r  \chi_{|q| \leq N^\alpha} (a^\dagger_{p+q} a^\dagger_{-p} a_q a_0 + \text{h.c.})+ \mathcal{E}_{[Q_4,\B_3]}
\end{align*} 
where 
\[ \begin{split}
2 \mathcal{E}_{[Q_4,\B_3]}=  \; &- \sum_{p,q,m, r} \hat V_N(r) \tilde \vphi_m  \chi_{|p+r| \leq N^\alpha}  \, a^\dagger_{m+p+r} a^\dagger_{-m} a^\dagger_q a_{q+r} a_p a_0 \\
& -  \sum_{p,q,m, r} \hat V_N(r) \tilde \vphi_m \chi_{|q| \leq N^\alpha} \, a^\dagger_{m+q} a^\dagger_{-m} a^\dagger_{p+r} a_{q+r} a_p a_0 \\
& +  \sum_{p,q,m, r}  \hat V_N(r) \tilde \vphi_m  \chi_{|q+r-m| \leq N^\alpha} \, a^\dagger_{p+r} a^\dagger_{-m} a^\dagger_{q} a_{q+r-m} a_p a_0 \\
& +  \sum_{p,q,m, r}  \hat V_N(r) \tilde \vphi_m \chi_{|p-m| \leq N^\alpha} \, a^\dagger_{p+r} a^\dagger_{-m} a^\dagger_{q} a_{q+r} a_{p-m} a_0 \\
& +  \sum_{p,q,m, r} \hat V_N(r) \tilde \vphi_{-q-r}  \chi_{|m| \leq N^\alpha} \, a^\dagger_{p+r} a^\dagger_{q} a^\dagger_{-q-r+m} a_{p} a_{m} a_0 \\
& +  \sum_{p,q,m, r} \hat V_N(r) \tilde \vphi_{-p} \chi_{|p-m| \leq N^\alpha} \, a^\dagger_{p+r} a^\dagger_{q} a^\dagger_{-m} a_{q+r} a_{p-m} a_0  + \text{h.c.} =: \; \sum_{i=1}^6 \mathcal{E}_i  .
\end{split}\]
For a parameter $\delta > 0$, we find 
\begin{align*}
\pm \mathcal{E}_1
%\leq \; &\delta \sum_{m,p,q,r} |\hat V_N(r)|^2 \chi_{|p+r|\leq N^\alpha} \| a_{m+p+r} a_{-m} %a_q \xi \|^2 \\ &+\delta^{-1} \sum_{m,p,q,r}  \chi_{|p+r|\leq N^\alpha} |\vphi_m|^2  \| a_{q+r} %a_p a_0 \xi \|^2 \\ 
&\lesssim \delta \|V_N\|_{\infty}^2  \|\chi_{|\cdot|\leq N^\alpha}\|_{1}   \N_+^{3} + \delta^{-1} \|\tilde \vphi\|_2^2  \| \chi_{|\cdot|\leq N^\alpha}\|_{1} \N^2_+ \cN_0 \, , \\ 
&\lesssim N^{-1+5\alpha/2} \cN_+^2 \end{align*}
where in the last step, we chose $\delta = N^{-\alpha/2}$ and we used $\cN_0 \leq N$ (and Lemma \ref{Lemmavphi}). To estimate $\mathcal{E}_2, \dots, \mathcal{E}_6$, we switch to position space. For arbitrary $\delta > 0$, we find 
\begin{align*}
\pm \mathcal{E}_2 
	&= \pm \int_{\Lambda^2}  \dx  \dy  V_N(x-y) \check{\chi}_{|\cdot|\leq N^{\alpha}}(z-y) a^\dagger_z a^\dagger(\check{\tilde \vphi}^z) a^\dagger_x a_x a_y a_0 + \text{h.c.} \\
	&\leq \delta \|\chi_{|\cdot|\leq N^{\alpha}}\|_2^2 \, Q_4 a^\dagger_0 a_0 + \delta^{-1} 
	\int_{\Lambda^2}  \dx  \dy   V_N(x-y)  a^\dagger_z a^\dagger(\check{\tilde \vphi}^z) a^\dagger_x a_x a(\check{\tilde \vphi}^z) a_z \\
	&\lesssim \delta N^{1+3\alpha} Q_4 + \delta^{-1} N^{-3-\alpha} \cN_+^3  \\
	&\lesssim N^{-\alpha/2} Q_4 + N^{-2+5\alpha/2} \N_+^3,
\end{align*}
where, in the last line, we fixed $\delta = N^{-1-7\alpha/2}$. Similarly, we find \begin{align*}
\pm \mathcal{E}_3 
	&= \pm \int_{\Lambda^2}  \dx  \dy   V_N(x-y) a^\dagger_xa^\dagger_y a^\dagger(\check{\tilde \vphi}^z) a(\check{\chi}_{|\cdot|\leq N^{\alpha}}^y) a_x a_0 + \text{h.c.} \\
%&\leq \delta Q_4 + \delta^{-1} \int V_N(x-y)   a^\dagger(\check{\tilde \vphi}^z) %a(\check{\chi}_{|\cdot|\leq N^{\alpha}}^y) a_x|^2 (a^\dagger_0 a_0) \\
	&\lesssim \delta Q_4 + \delta^{-1} N \|V_N\|_{1} \|\check \chi_{|\cdot|\leq N^{\alpha}}\|_{2}^2 \|\tilde \vphi\|_2^2 (\N_++1)^3 \\
	&\lesssim N^{-\alpha/2} Q_4 + N^{-2+5\alpha/2} (\N_+ + 1)^3,
\end{align*}
taking $\delta = N^{-\alpha/2}$. Furthermore, for an arbitrary $\xi \in \cF$, we have 
\begin{align*}
|\langle \xi,  \mathcal{E}_4 \xi \rangle| &= \Big| \int_{\Lambda^2}  \dx  \dy   V_N(x-y) \check{\chi}_{|\cdot|\leq N^{\alpha}}(x-z) \langle \xi, a^\dagger_xa^\dagger_z a^\dagger(\check{\tilde \vphi}^x) a_y a_z a_0 \xi \rangle  \Big| \\
&\leq \| \check{\chi}_{|.| \leq N^\alpha} \|_\infty \int _{\Lambda^2}  \dx  \dy V_N (x-y) \| a_x a_z a (\check{\tilde\vphi}^x) \xi \| \| a_y a_z a_0 \xi \|  \\
&\lesssim  \| \chi_{|.| \leq N^\alpha} \|_1 \| \tilde\vphi \|_2 \| V_N \|_1 \left[ \langle \xi , \cN^3_+ \xi \rangle + \langle \xi, \cN_+^2 \cN_0 \xi \rangle \right]  \\
&\lesssim N^{-1+5\alpha/2}  \langle \xi, \cN_+^2 \xi \rangle.
\end{align*} 
As for $\mathcal{E}_5$, we estimate 
\begin{align*}
\pm \mathcal{E}_5
	&= \pm \int_{\Lambda^2}  \dx  \dy   V_N(x-y) \check{\tilde \vphi}(y-z) a^\dagger_xa^\dagger_y a^\dagger_z (\N_++1)^{-1/2}(\N_++1)^{1/2} a_x  a(\check{\chi}_{|\cdot|\leq  N^{\alpha}}^z) a_0 + \text{h.c.} \\
	&\leq \delta Q_4 + \delta^{-1} \int_{\Lambda^2}  \dx  \dy  V_N(x-y) |\check{\tilde \vphi}(z-y)|^2  a_x^\dagger a^\dagger (\check{\chi}^z_{|.| \leq N^\alpha}) a_0^\dagger (\cN_+ + 1) a_0 a (\check{\chi}^z_{|.| \leq N^\alpha})  a_x \\
	&\lesssim \delta Q_4 + \delta^{-1} N \|V_N\|_{1} \|\check \chi_{|\cdot|\leq N^{\alpha}}\|_{2}^2 \|\tilde \vphi\|_2^2 (\N_++1)^3 \\
	&\lesssim N^{-\alpha/2} Q_4 + N^{-2+5\alpha/2} \N_+^3,
\end{align*}
choosing again $\delta = N^{-\alpha/2}$. By a simple change of variable, it is easy to check that $\mathcal{E}_6 = \mathcal{E}_5$. This concludes the proof of the lemma. 
\end{proof}

With the bounds from the last lemma, we can estimate the operator $\Gamma_3$, defined in (\ref{eq:Xi3}). 
\begin{lemma} \label{lm:Xi3-bd} 
We have 
\begin{equation}\label{eq:Xi3-bd} \begin{split}  \pm \Gamma_3 \lesssim \; &N^{-3\alpha/2} H_1 + N^{-\alpha/2} Q_4 + N^{3\alpha/2-1/2} (\cN_+ + 1)^{3/2} + N^{-1+5\alpha/2} (\cN_+ + 1)^2.  \end{split} \end{equation} 
\end{lemma} 
\begin{proof}
With Lemma \ref{H1B3}, we find, using the scattering equation (\ref{scattequ}),  
\[ \Gamma_3 =  \big[ H_1 + Q_4 , \B_3 \big] + Q_3 = \tilde{Q}_{3,1} + \tilde{Q}_{3,2} 
+ Q_3^> + \mathcal{E}_{[H_1,\B_3]} + \mathcal{E}_{[Q_4,\B_3]}
\]
with 
\begin{align}
\tilde Q_{3,1} & = - \sum_{p,q} p^2 \vphi_p \chi_{|p| \leq N^\alpha} \chi_{|q| \leq N^\alpha} \, a^\dagger_{p+q} a^\dagger_{-p} a_q a_0 + \text{h.c.}, \label{eq:tQ3}\\
\tilde Q_{3,2} &=  - \sum_{p,q, r} \hat{V}_N (p-r) \vphi_r \chi_{|r| \leq N^\alpha}  \chi_{|q| \leq N^\alpha} \, a^\dagger_{p+q} a^\dagger_{-p} a_q a_0 + \text{h.c.}, \\
Q_3^{>} & =  \sum_{p,q}  \hat V_N(p) \chi_{|q| >  N^\alpha} a^\dagger_{p+q} a^\dagger_{-p} a_q a_0 +\text{h.c.} \label{eq:Q3>}
\end{align}
It follows easily from Lemma \ref{Lemmavphi} that $\| \chi_{|p| \leq N^{\alpha}} p^2 \vphi_p\|_{2}\lesssim N^{-1+3\alpha/2}$; thus  
\[  \pm \tilde Q_{3,1} \lesssim N^{-1/2 + 3\alpha/2} (\N_++1)^{3/2}. \]
Denoting $\vphi^{<}_p = \vphi_p \chi_{|p| \leq N^{\alpha}}$, we write $\tilde{Q}_{3,2}$ in position space as 
\begin{align*}
\pm \tilde Q_{3,2} 
	&= \pm \int_{\Lambda^3}  \dx  \dy  \dz  \, V_N (x-y) \check{\vphi}^{<} (x-y) \check{\chi}_{|\cdot|\leq N^{\alpha}}(x-z) a^\dagger_x a^\dagger_y a_z a_0 + \text{h.c.}\\
	&\leq \delta Q_4 + \delta^{-1} N \| V_N^{1/2} \check\vphi^< \|_2^2 \|\chi_{|\cdot|\leq N^{\alpha}}\|_\infty^2  (\N_++1) \\
	&\lesssim N^{-\alpha/2} Q_4 + N^{-2+5\alpha/2} (\cN_+ + 1) \, ,
\end{align*} 
where we chose $\delta = N^{-\alpha/2}$ and we used the estimate 
\begin{equation}\label{eq:dGchi} \int_\Lambda a^\dagger (\check{\chi}_{|.| \leq N^\alpha}^x) \, a (\check{\chi}_{|.| \leq N^\alpha}^x) \, \dx  = d\Gamma (\chi^2_{|p| \leq N^\alpha}) \leq \| \chi_{|.| \leq N^\alpha} \|_\infty^2 \, \cN_+ \leq \cN_+ \, . \end{equation}  
Proceeding similarly, we find 
\begin{align*}
\pm Q_3^{>} 
	&= \pm \int_{\Lambda^3} \dx  \dy  \dz  \, V_N (x-y) \check{\chi}_{|\cdot|> N^{\alpha}}(x-z) a^\dagger_x a^\dagger_y a_z a_0 + \text{h.c.}\\
	&\lesssim \delta Q_4 + \delta^{-1} N^{1-2\alpha} \| V_N \|_1 H_1 \lesssim N^{-\alpha/2} Q_4 + N^{-3\alpha/2} H_1,
\end{align*}
where we took $\delta = N^{-\alpha/2}$ and we used that  
\[ \int_\Lambda a^\dagger(\check{\chi}_{|\cdot|> N^{\alpha}}^{x}) a(\check{\chi}_{|\cdot|> N^{\alpha}}^x) \dx = \sum_{|p| > N^\alpha} a_p^\dagger a_p \leq N^{-2\alpha} H_1 \, . \]
Combining the bounds for $\tilde{Q}_{3,1} , \tilde{Q}_{3,2}, Q_3^>$ with the estimates for $\mathcal{E}_{[H_1,\B_3]}, \mathcal{E}_{[Q_4,\B_3]}$ from Lemma \ref{H1B3}, we obtain (\ref{eq:Xi3-bd}). 
 \end{proof}
 
 In order to obtain similar bounds also for the integral, we also need a-priori control on growth of $H_1, Q_4$. 
 \begin{lemma} \label{lm:a-priB3} 
We have 
 \begin{align}
e^{-s \B_3} Q_4 e^{s \B_3} &\lesssim Q_4 + \N_+ + 1 + N^{-1+5\alpha/2} (\N_++1)^2, \label{eq:gronwall_Q4_B3}\\
e^{-s \B_3} H_1e^{s \B_3} &\lesssim H_1 + Q_4 + \N_+ +1 + N^{-1+3\alpha} (\N_++1)^2. \label{eq:gronwall_H1_B3}
\end{align}
 \end{lemma} 
\begin{proof}
For arbitrary $\xi \in \cF$, we define $f(s) =  \langle \xi, e^{-s\B_3} Q_4 e^{s \B_3} \xi \rangle$, so that 
\[ f' (s) = \langle \xi, e^{-s \B_3} [ Q_4 , \cB_3 ]  e^{s \B_3} \xi \rangle \, . \]
From Lemma \ref{H1B3}, we find
\[ [Q_4,\B_3] = \sum_{r,p,q} \hat V_N(p-r)\tilde \vphi_r  \chi_{|q| \leq N^\alpha} (a^\dagger_{p+q} a^\dagger_{-p} a_q a_0 + \text{h.c.})+ \mathcal{E}_{[Q_4,\B_3]}
. \]
where 
\[ \pm \mathcal{E}_{[Q_4,\B_3]}
 \lesssim N^{-\alpha/2} Q_4 + N^{-1+5\alpha/2} (\cN_+ + 1)^2 \, . \]
Switching to position space, we have 
\[ \begin{split}  \sum_{r,p,q} \hat V_N(p-r) &\tilde \vphi_r  \chi_{|q| \leq N^\alpha} (a^\dagger_{p+q} a^\dagger_{-p} a_q a_0 + \text{h.c.}) \\ &=  \sum_{r,p,q} \hat V_N(p-r)\tilde \vphi_r  \chi_{|q| \leq N^\alpha} (a^\dagger_{p+q} a^\dagger_{-p} a_q a_0 + \text{h.c.})\\
	&= \int_{\Lambda^2} \dx  \dy  \, V_N (x-y) \check{\tilde\vphi} (x-y) a^\dagger_x a^\dagger_y a( \check{\chi}_{|\cdot| \leq N^\alpha}^x) a_0 + \text{h.c.} \\
	&\lesssim Q_4 + \int \dx  \dy  \, V_N(x-y) |\check{\tilde\vphi} (x-y)|^2  a_0^\dagger a^\dagger (\check{\chi}^x_{|.| \leq N^\alpha}) a (\check{\chi}^x_{|.| \leq N^\alpha}) a_0 \\ &\lesssim Q_4 + N \| V_N \|_1 \| \tilde \vphi \|_\infty^2 \| \chi_{|.| \leq N^\alpha}^2 \|_\infty \lesssim Q_4 +  \N_+ ,
\end{split}\]
where we used Lemma \ref{Lemmavphi} and we argued as in (\ref{eq:dGchi}). We conclude that 
\[ \pm \big[ Q_4 , \B_3 \big] \lesssim Q_4 + \cN_+ + N^{-1+5\alpha/2} (\cN_+ + 1)^2. \]
Therefore, using Lemma \ref{Gron}, we find 
\[ f' (s) \lesssim f(s) + \langle \xi, \cN_+ \xi \rangle + N^{-1+5\alpha/2} \langle \xi , (\cN_+ + 1)^2 \xi \rangle. \]
By Gronwall, we obtain (\ref{eq:gronwall_Q4_B3}). 

To prove (\ref{eq:gronwall_H1_B3}), we proceed similarly. For $\xi \in \cF$, we define $g(s) =  \langle \xi, e^{-s\B_3} H_1 e^{s \B_3} \xi \rangle$, for any $|s| \leq 1$, which leads to 
\[ g' (s) = \langle \xi , e^{-\cB_3} H_1 e^{s\cB_3} \xi \rangle. \]
From Lemma \ref{H1B3}, we have 
\[ \big[ H_1, \B_3 \big] = \sum_{p,q,r} \hat{V}_N (p-r) (\delta_{0,r} + \vphi_r) \chi_{|p| > N^\alpha} \chi_{|q| \leq N^\alpha} a_{p+q}^\dagger a_{-p}^\dagger a_q a_0 + \text{h.c.} + \mathcal{E}_{[H_1,\B_3]}
 \]
where
\[ \pm \mathcal{E}_{[H_1,\B_3]}
 \leq N^{-3\alpha/2} H_1 + N^{-1+5\alpha/2} (\cN_+ + 1)^2 \,. \]
Writing $\chi_{|p| > N^\alpha} = 1 - \chi_{|p| \leq N^\alpha}$, we decompose
\[  \sum_{p,q,r} \hat{V}_N (p-r) (\delta_{0,r} + \vphi_r) \chi_{|p| > N^\alpha} \chi_{|q| \leq N^\alpha} a_{p+q}^\dagger a_{-p}^\dagger a_q a_0 + \text{h.c.} = \mathcal{E}_1 + \mathcal{E}_2 \]
where 
\[ \begin{split} \mathcal{E}_1 &= \pm \int_{\Lambda^2} \dx  \dy  V_N (x-y) (1 + \check{\vphi}) (x-y) a_x^\dagger a_y^\dagger a(\check{\chi}_{|.| \leq N^\alpha}^x) a_0 + \text{h.c.}  \\
&\lesssim  Q_4 +  N \|V_N\|_{1} \|(1+\check \vphi)\|_{\infty}^2 \|\chi_{|\cdot|\leq N^\alpha}\|_{\infty}^2 \N_+ \lesssim Q_4 + \N_+ 
\end{split}\] 
and  
\[ \begin{split} 
\pm \mathcal{E}_2 &= \pm \sum_{p,q,r} \hat{V}_N (p-r) (\delta_{0,r} + \vphi_r ) \chi_{|p| \leq N^\alpha} \chi_{|q| \leq N^\alpha} a_{p+q}^\dagger a_{-p}^\dagger a_q a_0 + \text{h.c.}  
\\ & \lesssim \| V_N \|_1 \| (1+ \check{\vphi}) \|_\infty \Big[ \delta \cN_+^2 + \delta^{-1} N \| \chi_{|.| \leq N^\alpha} \|_2^2 \cN_+ \Big] \\ &\lesssim N^{-1} \Big[ \delta \cN_+^2 + \delta^{-1} N^{1+3\alpha} \cN_+ \Big]  \lesssim \cN_+ + N^{-1+3\alpha} \cN_+^2\, ,
\end{split} \]
choosing in the last step $\delta = N^{3\alpha}$. Thus, with Lemma \ref{Gron}, we find 
\[ g' (s) \lesssim f(s) + g(s) + \langle \xi, \cN_+ \xi \rangle + N^{-1+3\alpha} \langle \xi, (\cN_+ + 1)^2 \xi \rangle. \]
With (\ref{eq:gronwall_Q4_B3}) and applying Gronwall's lemma, we obtain (\ref{eq:gronwall_H1_B3}). 
\end{proof} 

In the next lemma, we control the contribution on the r.h.s. of (\ref{BHm3}), arising from the commutator $[Q_3, \cB_3]$. 
\begin{lemma}\label{lemma:Q3B3}
We have
$$ 
    \int_0^1\int_s^1 e^{-t\B_3} [Q_3,B_3]e^{t\B_3} \dt\ds =  2 (8\pi \mathfrak{a}_N - \hat{V} (0)) \cN_+  + \mathcal{E}_{[Q_3,\B_3]}, $$
    with 
    $$ \pm  \mathcal{E}_{[Q_3,\B_3]} \lesssim  N^{-2\alpha} H_1 + N^{-\alpha/2} Q_4 + \big[ N^{-\alpha/2} + N^{-1+\alpha} \big] (\cN_+ + 1)  + N^{-1+5\alpha/2} (\N_++1)^2. 
$$
\end{lemma}
\begin{proof}
We compute
\begin{equation}\label{eq:IandII}
\begin{split} 
[Q_3,\B_3] &= \sum_{p,q,r,s%,p+q,r+s
} \hat V_N(p) \tilde \vphi_r \chi_{|s|\leq N^\alpha} [a^\dagger_{p+q} a^\dagger_{-p}a_q a_0 + a^\dagger_0 a^\dagger_q a_{-p} a_{p+q}, a^\dagger_{r+s}a^\dagger_{-r} a_sa_0 - a^\dagger_0 a^\dagger_sa_{-r} a_{r+s}] \\
	&=\sum_{p,q,r,s%,p+q,r+s
} \hat V_N(p) \tilde \vphi_r \chi_{|s|\leq N^\alpha} [a^\dagger_{p+q} a^\dagger_{-p}a_q a_0, a^\dagger_{r+s}a^\dagger_{-r} a_sa_0] + \text{h.c.} \\
	&\quad + \sum_{p,q,r,s%,p+q,r+s
} \hat V_N(p) \tilde \vphi_r \chi_{|s|\leq N^\alpha} [a^\dagger_0 a^\dagger_sa_{-r} a_{r+s},a^\dagger_{p+q} a^\dagger_{-p}a_q a_0] + \text{h.c.} =: (\textrm{I}) + (\textrm{II}).
\end{split}\end{equation}
% ?? The estimate on the second term on the left is a corollary of Lemma \ref{H1B3} and Lemma \ref{Gamma3}.

We start by estimating the term $(\textrm{I})$. With the canonical commutation relations, we obtain 
\begin{equation}\label{Q3B3_1}
\begin{split} 
(\textrm{I}) = \; &\sum_{p,q,r}\hat V_N(p) \left[ \tilde \vphi_r \chi_{|q-r|\leq N^\alpha} a^\dagger_{p+q} a^\dagger_{-p}a^\dagger_{-r} a_{q-r} +
\tilde \vphi_q \chi_{|r|\leq N^\alpha} a^\dagger_{p+q} a^\dagger_{-p}a^\dagger_{-q+r} a_{r}  \right. \\ &\hspace{1cm} - \left.
\tilde \vphi_r \chi_{|q+p|\leq N^\alpha} a^\dagger_{r+p+q} a^\dagger_{-r}a^\dagger_{-p} a_{q} -
\tilde \vphi_r \chi_{|p|\leq N^\alpha} a^\dagger_{r-p} a^\dagger_{-r}a^\dagger_{p+q} a_{q}\right] a_0a_0 + \text{h.c.}  \\ =: \; &(\textrm{I})_a + (\textrm{I})_b + (\textrm{I})_c + (\textrm{I})_d. 
\end{split} \end{equation} 
To estimate the first term, we rewrite it in position space. We find 
\begin{align*}
\pm (\textrm{I})_a &= \pm \int_{\Lambda^2} \dx  \dy  \, V_N(x-y) a^\dagger_x a^\dagger_y  a^\dagger(\check {\tilde \vphi}^x) a(\check{ \chi}^x_{|\cdot| \leq N^\alpha}) a_0 a_0 + \text{h.c.} \\
	&\leq \delta Q_4 + \delta^{-1} \int_{\Lambda^2} \dx  \dy  \, V_N(x-y)  a^\dagger (\check{ \chi}^x_{|\cdot| \leq N^\alpha}) a(\check {\tilde \vphi}^x)a^\dagger(\check {\tilde \vphi}^x)  a(\check{ \chi}^x_{|\cdot| \leq N^\alpha}) a^\dagger_0a^\dagger_0 a_0 a_0 \\
	&\lesssim \delta Q_4 + \delta^{-1} N^2 \|\tilde \vphi\|_2^2 \|V_N\|_1 \|\chi_{|\cdot| \leq N^\alpha}\|_{\infty}^2 (\N_++1)^2 \\
	&\lesssim \delta Q_4 + \delta^{-1} N^{-\alpha-1} (\N_++1)^2 \leq N^{-\alpha/2} Q_4 + N^{-1-\alpha/2} (\cN_+ + 1)^2 
\end{align*}
where we used that $\N_0 \leq N$ and the bound (\ref{eq:dGchi}) and, in the last step, we set $\delta = N^{-\alpha/2}$. The second term in \eqref{Q3B3_1} is dealt with similarly. We obtain 
\begin{align*}
\pm (\textrm{I})_b &= \pm \int_{\Lambda^3} \dx  \dy  \dz  \,  V_N(x-y)   \check {\tilde \vphi}(x-z) \, a^\dagger_x a^\dagger_y  a^\dagger_z  a(\check{ \chi}^z_{|\cdot| \leq N^\alpha}) a_0 a_0 + \text{h.c.} \\
	&\leq \delta Q_4 + \delta^{-1}\int_{\Lambda^3} \dx  \dy  \dz  \, V_N(x-y)  |\check {\tilde \vphi}(x-z)|^2  a^\dagger(\check{ \chi}^z_{|\cdot| \leq N^\alpha}) (\N_++1)a(\check{ \chi}^z_{|\cdot| \leq N^\alpha}) a^\dagger_0a^\dagger_0a_0a_0 \\
	& \lesssim \delta Q_4 + \delta^{-1} N^2 \|\tilde \vphi\|_2^2 \|V_N\|_{1} \|\chi_{|\cdot| \leq N^\alpha}\|_{\infty}^2 (\N_++1)^2 \\
	&\lesssim  N^{-\alpha/2} Q_4 + N^{-1-\alpha/2} (\cN_+ + 1)^2 
\end{align*}
choosing again $\delta= N^{-\alpha/2}$. As for the third term in \eqref{Q3B3_1}, we bound it, for an arbitrary $\delta > 0$, with Cauchy-Schwarz by   
\begin{align*}
\pm (\textrm{I})_c&\lesssim \delta \sum_{p,q,r}|\hat V_N(p)| a^\dagger_{r+p+q} a^\dagger_{-r}a^\dagger_{-p} (\N_+ +1)^{-1}a_{-p} a_{-r} a_{r+p+q} \\
	&\qquad + \delta^{-1} \sum_{p,q,r} |\hat V_N(p)| |\tilde \vphi_r|^2 \chi_{|q+p|\leq N^\alpha} a^\dagger_{q}(\N_++1)a_{q}a^\dagger_0a^\dagger_0a_0a_0 \\
	&\lesssim (\delta \|\hat V_N\|_{\infty} + \delta^{-1} N^2 \|\tilde \vphi\|_{2}^2 \|\hat V_N \ast \chi_{|\cdot|\leq N^{\alpha}}\|_{\infty})(\N_++1)^2  \\
	&\lesssim N^{- 1+5\alpha/2} (\N_+ + 1)^2
\end{align*}
where at the end we took $\delta = N^{-\alpha}$ and we used $\|\hat V_N \ast \chi_{|\cdot|\leq N^{\alpha}}\|_{\infty} \leq \|\hat V_N\|_{\infty} \|\chi_{|\cdot|\leq N^{\alpha}}\|_{1} \lesssim N^{-1+3\alpha}$. The last term in \eqref{Q3B3_1} can be bounded, again by Cauchy-Schwarz, by 
\begin{align*}
\pm (\textrm{I})_d &\lesssim  \delta (\cN_+ + 1)^2 + \delta^{-1} N^2 \| \hat{V}_N \|_\infty^2 \| \tilde\vphi \|_2^2 \| \chi_{|.| \leq N^\alpha} \|_1 (\cN_+ + 1)^2  \lesssim N^{-1+\alpha} (\cN_+ + 1)^2 
\end{align*}
where we used $\delta = N^{\alpha}$. 

Let us now consider the term $(\textrm{II})$ in (\ref{eq:IandII}). We write 
\begin{equation}\label{eq:IIaIIb} \begin{split} (\textrm{II}) &= \sum_{p,q,r,s} \hat{V}_N (p) \tilde\vphi_r \chi_{|s| \leq N^\alpha}  \\ &\hspace{1cm} \times 
\Big\{ a_0^\dagger a_s^\dagger \big[ a_{-r} a_{r+s} , a^\dagger_{p+q} a^\dagger_{-p} \big] a_q a_0 + a_{p+q}^\dagger a_{-p}^\dagger \big[ a_0^\dagger a_s^\dagger , a_q a_0 \big] a_{-r} a_{r+s} \Big\}  + \text{h.c.} \\ &=: (\textrm{II})_a + (\textrm{II})_b \,. \end{split} \end{equation} 
With 
\[ a_{p+q}^\dagger a_{-p}^\dagger \big[ a_0^\dagger a_s^\dagger , a_q a_0 \big] a_{-r} a_{r+s}  = - \delta_{sq} \, a_{p+q}^\dagger a_{-p}^\dagger a_0^\dagger a_0 a_{-r} a_{r+s} -  a_{p+q}^\dagger a_{-p}^\dagger a_q a_s^\dagger a_{-r} a_{r+s} \]
we obtain 
\[ \begin{split} (\textrm{II})_b = \; &-\sum_{p,q,r} \hat{V}_N (p) \tilde\vphi_r \chi_{|q| \leq N^\alpha} \, a_{p+q}^\dagger a_{-p}^\dagger a_0^\dagger a_0 a_{-r} a_{r+q} \\ 
&- \sum_{p,q,r,s} \hat{V}_N (p) \tilde\vphi_r \chi_{|s| \leq N^\alpha} \, a_{p+q}^\dagger a_{-p}^\dagger a_s^\dagger a_q a_{-r} a_{r+s} \\ &- \sum_{p,r,q} \hat{V}_N (p) \tilde\vphi_r \chi_{|q| \leq N^\alpha} \, a_{p+q}^\dagger a^\dagger_{-p} a_{-r} a_{r+q} =: (\textrm{II})_{b1} +  (\textrm{II})_{b2} +  (\textrm{II})_{b3}. \end{split} \]
We can bound $(\textrm{II})_{b3}$ switching to position space. We find 
\[ \begin{split} \pm  (\textrm{II})_{b3}  &= \pm \int_{\Lambda^4} \dx  \dy  \du  \dv  \, V_N (x-y) \check{\chi}_{|.| \leq N^\alpha} (x-u) \check{\tilde \vphi} (u-v) a_x^\dagger a_y^\dagger a_u a_v + \text{h.c.} \\
&\lesssim \delta \| \check{\tilde \vphi} \|_2^2 \, Q_4 + \delta^{-1} \| \check{\chi}_{|.| \leq N^\alpha} \|_2^2 \,  (\cN_+ + 1)^2  \lesssim N^{-\alpha/2} Q_4 + N^{-3 +5 \alpha/2} (\cN_+ + 1)^2  \, . \end{split}  \]
The term $(\textrm{II})_{b1}$ can be bounded analogously, but it contains an additional factor $\cN_0 = a_0^\dagger a_0 \leq N$. Thus
\[   \pm  (\textrm{II})_{b1} \leq  N^{-\alpha/2} Q_4 + N^{-1 +5 \alpha/2} (\cN_+ + 1)^2.   \]
Also $(\textrm{II})_{b1}$ can be bounded in position space. We obtain 
\[ \begin{split} \pm  (\textrm{II})_{b2}  &= \pm \int_{\Lambda^3} \dx  \dy  \du  \, V_N (x-y) \,  a_x^\dagger a_y^\dagger a^\dagger (\check{\chi}_{|.| \leq N^\alpha}^u) a_x a (\check{\tilde \vphi}^u) a_u + \text{h.c.} \\ &\leq \delta \| \check{\chi}_{|.| \leq N^\alpha} \|_2^2 \, Q_4 \cN_+ + \delta^{-1} \| V_N \|_1 \| \check{\tilde \vphi} \|_2^2 (\cN_+ + 1)^3 \lesssim N^{-\alpha/2} Q_4 + N^{-1+5\alpha/2} (\cN_+ + 1)^2  \end{split} \]
where we chose $\delta = N^{-1-7\alpha/2}$.

Let us now consider the term $(\textrm{II})_a$, defined on the r.h.s. of (\ref{eq:IIaIIb}). 
With 
\[ \begin{split} 
a_0^\dagger &a_s^\dagger \big[ a_{-r} a_{r+s} , a^\dagger_{p+q} a^\dagger_{-p} \big] a_q a_0 \\ &= a_0^\dagger a_s^\dagger \Big\{ \delta_{-r,p+q} a_{-p}^\dagger a_{r+s} + \delta_{r,p} a^\dagger_{p+q} a_{r+s} + \delta_{r+s, p+q} a_{-r} a_{-p}^\dagger + \delta_{r+s,-p} a_{-r} a_{p+q}^\dagger \Big\} a_q a_0 \end{split} \]
we obtain, rearranging terms in normal order (with appropriate changes of variables)
\begin{equation}\label{eq:IIa1-4} 
\begin{split} 
(\textrm{II})_a = \; &\sum_{p,q} \big( \hat{V}_N (p) + \hat{V}_N (p-q) \big) \tilde\vphi_p \chi_{|q| \leq N^\alpha} a_0^\dagger a_q^\dagger a_q a_0 \\ &+
4 \sum_{p,q,s} \hat{V}_N (p) \tilde\vphi_{p+q} \chi_{|s| \leq N^\alpha} \, a_0^\dagger a_s^\dagger a_{-p}^\dagger a_{-p-q+s} a_q a_0 + \text{h.c.} \\ 
=: \; &2 \sum_{p,q} \big( \hat{V}_N (p) + \hat{V}_N (p-q) \big) \tilde\vphi_p \chi_{|q| \leq N^\alpha} a_0^\dagger a_q^\dagger a_q a_0 + (\textrm{II})_{a1} 
\end{split} \end{equation} 
where we can bound, with $\cN_0 = a_0^\dagger a_0 \leq N$, 
\[ \begin{split} \pm  (\textrm{II})_{a1} \lesssim N \| \hat{V}_N \|_\infty \Big[ \delta \| \tilde \vphi \|_2^2 + \delta^{-1} \| \chi_{|.| \leq N^\alpha} \|_2^2 \Big] \cN_+^2 \lesssim N^{-1+\alpha} (\cN_+ + 1)^2
 \end{split} \]
choosing $\delta = N^{1-\alpha}$. Collecting all the estimates we proved so far, we conclude from (\ref{eq:IandII}) that 
\[ \big[ Q_3, \cB_3 \big] = 2 \sum_{p,q} \big( \hat{V}_N (p) + \hat{V}_N (p-q) \big) \tilde\vphi_p \chi_{|q| \leq N^\alpha} a_0^\dagger a_q^\dagger a_q a_0 +  \mathcal{E}_1 \]
where 
\begin{equation}\label{eq:Xi3-1} \pm  \mathcal{E}_1 \lesssim N^{-\alpha/2} Q_4 + N^{-1+5\alpha/2} (\cN_+ + 1)^2. \end{equation} 
At the expenses of adding an additional small error to the r.h.s. of (\ref{eq:Xi3-1}), in the main term we can replace $a_0^\dagger a_0 = N - \cN_+$ by a factor of $N$, since
\[ \pm \sum_{p,q} (\hat{V}_N (p) + \hat{V}_N (p-q) ) \tilde\vphi_p \chi_{|q| \leq N^\alpha} a_q^\dagger \cN_+ a_q \lesssim N^{-1} \| \tilde \vphi \|_1 (\cN_++1)^2 \lesssim N^{-1} (\cN_+ + 1)^2. \]
Moreover, from 
\begin{multline}
\pm N\sum_{p,q} \left( \hat V_N(p-q)- \hat V_N(p)\right)\tilde \vphi_p \chi_{|q| \leq N^\alpha} a^\dagger_q a_q \\
\lesssim \sum_{p,q} | q/N| \| \nabla \hat V\|_{\infty} |\tilde \vphi_p| \chi_{|q| \leq N^\alpha} a^\dagger_q a_q\lesssim \|\tilde \vphi\|_1 N^{-1+\alpha} \N_+ \lesssim N^{-1+\alpha} \N_+,  %\lesssim  \N_+\lesssim o(1).\text{(o or O?)}
\end{multline}
we arrive at 
\[ \big[ Q_3, \cB_3 \big] = 4N \sum_{p} \hat{V}_N (p) \tilde\vphi_p   \sum_{|q| \leq N^\alpha} a_q^\dagger a_q  + \mathcal{E}_2  \]
where 
\[ \pm \mathcal{E}_2 \lesssim N^{-\alpha/2} Q_4 + N^{-1+\alpha} \cN_+ + N^{-1+5\alpha/2}  (\cN_+ + 1)^2. \]
Conjugating with $e^{t \B_3}$ and integrating over $t,s$, we obtain, with the help of Lemma \ref{lemma:gronwall_B3} and Lemma \ref{lm:a-priB3} (and of the observation that the first estimate in Lemma \ref{lemma:gronwall_B3} also holds, if we replace, on the l.h.s. $\cN_+$ by $\sum_{|p| \leq N^\alpha} a_p^\dagger a_p$), 
\[ \int_0^1 \int_s^1 e^{-t\B_3} [Q_3, \B_3] e^{t \B_3} \dt \ds = 2N \sum_q \hat{V}_N (q) \tilde\vphi_q \sum_{|p| \leq N^\alpha} a_p^\dagger a_p + \mathcal{E}_3 \]
where 
\[ \pm \mathcal{E}_3 \lesssim N^{-\alpha/2} (Q_4 + \cN_+ + 1) + N^{-1+\alpha} \cN_+ + N^{-1+5\alpha/2} (\cN_+ + 1)^2\, . \]
The claim now follows from (\ref{eq:approx_aN}) 
% the estimate {\bf [we need to show this]} 
%\[ \Big| N \sum_q \hat{V}_N (q) \tilde \vphi_q  - (8\pi \mathfrak{a}_N - \hat{V} (0)) \Big| \lesssim %N^{\alpha-1} \]
and the observation that  
\[  \cN_+  - \sum_{|p| \leq N^\alpha} a_p^\dagger a_p  = \sum_{|p| \geq N^\alpha} a_p^\dagger a_p \leq N^{-2\alpha} H_1 \, . \]
\end{proof}  

Finally, we consider the conjugation of the operator $\tilde{Q}_2$, defined in (\ref{tQ2}). 
\begin{lemma}
	\label{lemma:Q2_B3}
We have 
\begin{equation}
	\label{eq:Q2_B3}
e^{-\B_3} \tilde Q_2 e^{\B_3}= \tilde{Q}_2  + \mathcal{E}_{\tilde Q_2}
\end{equation} with 
\[
 \pm  \mathcal{E}_{\tilde Q_2}  \lesssim N^{-1/2+\alpha} (\N_++1)^{3/2} + N^{-3/2+\alpha} (\N_++ 1)^{5/2}. \]
\end{lemma}

\begin{proof}
We have   
\[ \begin{split}  e^{-\B_3} \tilde Q_2 e^{\B_3} - \tilde Q_2   &= \int_0^1 e^{-s\B_3}[ \tilde Q_2,\B_3] e^{s\B_3}\ds \\ 
    &= \frac{4\pi \mathfrak{a}_N}{N}  \sum_{|r| \leq N^\alpha} \int_0^1 e^{-s\B_3}[ a^\dagger_r a^\dagger_{-r} a_0 a_0 + \text{h.c.}, \B_3] e^{s\B_3}\ds.
    \end{split}\] 
    We compute the commutator 
    \begin{align}
    	[a^\dagger_r a^\dagger_{-r} a_0 a_0 &+ \text{h.c.} , a_{p+q}^\dagger a_{-p}^\dagger a_q a_0 - \text{h.c.} ] \\ &=[a^\dagger_r a^\dagger_{-r}a_0a_0,a^\dagger_{p+q} a^\dagger_{-p} a_q a_0]+[a^\dagger_0 a^\dagger_0 a_r a_{-r},a^\dagger_{p+q} a^\dagger_{-p} a_q a_0]+\text{h.c.}\nonumber \\
	&= 2\Big\{ a^\dagger_0 a^\dagger_0 \big( \delta_{p+q,r} a^\dagger_{-p} a_{-p-q} +\delta_{p,r} a^\dagger_{p+q} a_{p} \big) a_q a_0 \nonumber\\
	&\qquad\qquad\qquad -\delta_{r,q}a^\dagger_{p+q} a^\dagger_{-p} a^\dagger_{-q} a_0 a_0 a_0- a^\dagger_{p+q} a^\dagger_{-p} a^\dagger_{0} a_{q} a_r a_{-r}\Big\} +\text{h.c.} \label{eq:comm_lemm_tQ2B3}
    \end{align}
Hence, we obtain
\[ \begin{split} 
\frac{4\pi \mathfrak{a}_N}{N} \sum_{|r| \leq N^\alpha} \big[ a_r^\dagger &a_{-r}^\dagger a_0 a_0  + \text{h.c.} , \B_3 \big] \\ =\; &\frac{8\pi \mathfrak{a}_N}{N} \sum_{p,q : |p+q| \leq N^\alpha} \tilde\vphi_p \chi_{|q| \leq N^\alpha}  a_0^\dagger a_0^\dagger a_{-p}^\dagger a_{-p-q} a_q a_0 \\ 
  &- \frac{4\pi \mathfrak{a}_N}{N} \sum_{p, |q| \leq N^\alpha} \tilde \vphi_p \chi_{|q| \leq N^\alpha} a_{p+q}^\dagger a_{-p}^\dagger a_{-q}^\dagger a_0 a_0 a_0 \\   &- \frac{4\pi \mathfrak{a}_N}{N} \sum_{p, q, |r| \leq N^\alpha} \tilde \vphi_p \chi_{|q| \leq N^\alpha} a_{p+q}^\dagger a_{-p}^\dagger a_{0}^\dagger a_q a_r a_{-r}   =: (\textrm{I}) + (\textrm{II}) + (\textrm{III}). 
  \end{split} \]
With Cauchy-Schwarz and using the bounds from Lemma \ref{Lemmavphi}, 
we can bound
\[ \begin{split} 
\pm (\textrm{I}) &\lesssim N^{-1/2-\alpha/2}  (\cN_+ + 1)^{3/2} , \\  \pm (\textrm{II}) &\lesssim N^{-1/2+\alpha} (\cN_+ + 1)^{3/2} , \\ \pm (\textrm{III}) &\lesssim N^{-3/2+\alpha}  (\cN_+ + 1)^{5/2}. \end{split} \]
 The claim now follows with Lemma \ref{lemma:gronwall_B3}. 
  \end{proof} 

We are now ready to conclude the proof of Prop. \ref{propB3}. 
\begin{proof}[Proof of Prop. \ref{propB3}] 
Recall from (\ref{BHm3}) that 
 \[ \begin{split} 
e^{-\B_3} e^{-\B_2} H_N e^{\B_2} e^{\B_3} = \; & 4 \pi \aN (N-1) + \sum_{|p|\leq N^{\alpha}} \frac{(4\pi \aN)^2}{p^2} +  H_1 + Q_4  + e^{-\B_3}(\tilde H_2 + \tilde Q_2 + \mathcal E_{\B_2})e^{\B_3} \\ &+ \int_0^1 e^{-t\B_3}  \Gamma_3 e^{t\B_3} \dt + \int_0^1 \int_s^1 e^{-t{\B_3}}[Q_3,{\B_3}]e^{t{\B_3}} \dt \ds \end{split} \] 
where
\[ \pm \mathcal{E}_{\B_2} \leq N^{-\alpha/2} Q_4 + \big[ N^{-\alpha/2} + N^{-1+5\alpha/2} \big]  (\cN_+ + 1)  + N^{-1+\alpha} \cN_+^2 + N^{-2} H_1. \]
With Lemma \ref{lemma:gronwall_B3} and Lemma \ref{lm:a-priB3}, this also implies that 
\[ \begin{split}  \pm e^{-\B_3} \mathcal{E}_{\B_2} e^{\B_3} \lesssim \; &N^{-\alpha/2} Q_4 + N^{-2} H_1 + \big[ N^{-\alpha/2} + N^{-1+5\alpha/2} \big] (\cN_+ + 1) \\ &+ N^{-1+2\alpha} (\cN_+ + 1)^2.
\end{split}  \]
Applying the first bound in Lemma \ref{lemma:gronwall_B3} to the operator $\tilde H_2 = (2\hat{V} (0) - 8 \pi \mathfrak{a}_N) \cN_+$, defined in (\ref{tH2}), we obtain 
\[ \pm \Big[ e^{-\B_3} \tilde H_2 e^{\B_3} -  \tilde H_2 \Big] \lesssim N^{-\alpha/2} (\cN_+ + 1). \]
Combining Lemma \ref{lm:Xi3-bd} with Lemma \ref{lm:a-priB3} we obtain 
\[ \begin{split}  \int_0^1 e^{-t \B_3} &\Gamma_3 e^{t \B_3} \dt \\ \lesssim \; & N^{-3\alpha/2} H_1 + N^{-\alpha/2} Q_4 + N^{-\alpha/2} (\cN_+ + 1)  \\ &+ N^{3\alpha/2-1/2} (\cN_+ + 1)^{3/2} + N^{-1+5\alpha/2}  (\cN_+ + 1)^2.
 \end{split} \]
Together with the bounds in Lemmas \ref{lemma:Q3B3} and \ref{lemma:Q2_B3} and with the observation that 
\[ 8\pi \mathfrak{a}_N \sum_{|p| > N^\alpha} a_p^\dagger a_p \leq N^{-2\alpha} H_1 \]
 we conclude the proof of Prop. \ref{propB3}. 
%\[ \begin{split} 
%e^{-\B_3} e^{-\B_2} H_N e^{\B_2} e^{\B_3} = \; & 4 \pi \aN (N-1) + \sum_{|p|\leq %N^{\alpha}} \frac{(4\pi \aN)^2}{p^2} \\ &+ \sum_p (p^2 + 8\pi \mathfrak{a}_N) a_p^\dagger a_p %+ \sum_{|p| \leq N^\alpha} 4\pi \mathfrak{a}_N \big[ a_p^\dagger a_{-p}^\dagger \frac{a_0 a_0}
%{N} + \text{h.c.}\big]   + Q_4 + \mathcal{E}_3 \end{split} \]
%where (using the assumption $\alpha < \alpha < 4 \alpha$, $0 < \alpha < 1/3$) 
%\[ \begin{split} \pm \mathcal{E}_3 \lesssim \; & \big[ N^{\alpha/2-2\alpha} + N^{-2\alpha} \big] %H_1 + \big[ N^{\alpha/2-2\alpha} + N^{-\alpha/2} \big] Q_4 + \big[ N^{\alpha/2-2\alpha}  + %N^{-\alpha/2}  \big] (\cN_+ + 1)  \\ &+ N^{(3\alpha-1)/2} (\cN_+ + 1)^{3/2} + \big[ N^{-1-
%\alpha/2+3\alpha} + N^{-1+7\alpha/2-2\alpha} \big] (\cN_+ + 1)^2   \end{split} \]
%This concludes the proof of Prop. \ref{propB3}. 
\end{proof}

\section{Diagonalization of quadratic Hamiltonian}
\label{sec:B4}

From Prop. \ref{propB3}, we find 
\begin{equation}\label{eq:diag0}
\begin{split} 
e^{-\B_3} &e^{-\B_2} H_N  e^{\B_2} e^{\B_3}
	\\ &= 4 \pi \aN (N-1) + \frac{1}{4}\sum_{|p|\leq N^{\alpha}} \frac{(8\pi \aN)^2}{p^2} + \sum_{|p| > N^\alpha} p^2 a_p^\dagger a_p + Q_4 \\
	&\quad +  \sum_{|p| \leq N^\alpha} (p^2 + 8\pi \aN ) a^\dagger_p \frac{a_0 a_0^\dagger}{N} a_p + \frac{1}{2}\sum_{|p| \leq N^\alpha} 8 \pi \aN  [a^\dagger_p a^\dagger_{-p} \frac{a_0 a_0 }N+\text{h.c.}]  + \mathcal{E},
\end{split} 
\end{equation}
with an error $\mathcal{E}$ satisfying the bound (\ref{eq:EB3}). Here we used the observation that, on the sector $\{ \cN = N \}$, we can write 
\[\begin{split}  \sum_{|p| \leq N^\alpha} (p^2 + 8\pi \aN) a_p^\dagger a_p &= \sum_{|p| \leq N^\alpha} (p^2 + 8\pi \aN) a_p^\dagger \frac{a_0^\dagger a_0 + \cN_+ + 1}{N} a_p \\ &= \sum_{|p| \leq N^\alpha} (p^2 + 8\pi \aN) a_p^\dagger \frac{a_0 a_0^\dagger}{N} a_p + \frac{1}{N} \sum_{|p| \leq N^\alpha} (p^2 + 8\pi \aN) a_p^\dagger \, \cN_+ a_p \end{split} \] 
where the term
\[ \frac{1}{N} \sum_{|p| \leq N^\alpha} (p^2 + 8\pi \aN) a_p^\dagger \, \cN_+ a_p \lesssim N^{2\alpha-1} (\cN_+ + 1)^2 \]
can be absorbed on the r.h.s. of (\ref{eq:EB3}). 

In this section, we are going to diagonalize the operator on the last line of (\ref{eq:diag0}). Inspired by Bogoliubov theory (on states with $a_0, a_0^\dagger \sim \sqrt{N}$, this operator is approximately quadratic), we define, for $|p| \leq N^\alpha$, the coefficients \[ \tau_p = - \frac{1}{4} \log \left[ 1 + \frac{16\pi \aN}{p^2} \right] \] 
so that 
\[ \tanh (2 \tau_p) = - \frac{8\pi\aN}{p^2 + 8\pi \aN} \, . \] 
We also introduce the notation $\gamma_p = \cosh \tau_p$ and $\nu_p = \sinh \tau_p$. 
\begin{lemma}\label{lm:tau} 
We have the pointwise bound $\gamma_p \lesssim 1$ and $\tau_p, \nu_p \lesssim \chi_{|p| \leq N^\alpha} /p^2$. Moreover, 
\[ \begin{split} \| \tau \|_\infty &\leq \| \tau \|_2 \lesssim 1 ,  \qquad \| \nu \|_\infty \leq \| \nu \|_2 \lesssim 1, \qquad  \| \gamma - 1 \|_\infty \leq \| \gamma - 1 \|_2 \lesssim 1 \end{split} \]
and 
\[ \begin{split} \| \check{\tau} \|_\infty  &\leq \| \tau \|_1 \lesssim N^\alpha  , \qquad  \| \check{\nu} \|_\infty \leq \| \nu \|_1 \lesssim N^\alpha. \end{split} \]
\end{lemma} 

With these coefficients, we can write 
\[ \begin{split} 
\sum_{|p| \leq N^\alpha} (p^2 &+ 8\pi \aN ) a^\dagger_p \frac{a_0 a_0^\dagger}{N} a_p + \frac{1}{2}\sum_{|p| \leq N^\alpha} 8 \pi \aN  [a^\dagger_p a^\dagger_{-p} \frac{a_0 a_0 }N+\text{h.c.}]  \\ = \; & \sum_{|p| \leq N^\alpha} \sqrt{|p|^4 + 16 \pi \aN p^2} \, \big(\gamma_p a_p^\dagger \frac{a_0}{\sqrt{N}} + \nu_p \frac{a_0^\dagger}{\sqrt{N}} a_{-p} \big) \big( \gamma_p \frac{a_0^\dagger}{\sqrt{N}} a_p + \nu_p a_{-p}^\dagger \frac{a_0}{\sqrt{N}} \big) \\ &-\frac{1}{2} \sum_{|p| \leq N^\alpha} \big[ p^2 + 8\pi \aN - \sqrt{|p|^4 + 16\pi \aN p^2} \big] + \delta \end{split} \]
with an error $\delta$ satisfying 
\[ \pm \delta \lesssim N^{\alpha-1} (\cN_+ + 1) \, . \]
Here, we used the relations 
\[ \begin{split} 
\gamma^2_p + \nu_p^2 &= \cosh (2\tau_p) = \frac{1}{\sqrt{1- \tanh^2 (2\tau_p)}} = \frac{p^2 + 8\pi \aN}{\sqrt{|p|^4 + 16\pi \aN p^2}}\, , \\
2\gamma_p \nu_p &= \sinh (2\tau_p) = \frac{\tanh (2\tau_p)}{\sqrt{1- \tanh^2 (2\tau_p)}} = \frac{8\pi \aN}{\sqrt{|p|^4 + 16\pi \aN p^2}} \, , \\
\nu^2_p &= \frac{1}{2} \big[ \cosh (2\tau_p) -1 \big] = \frac{1}{2} \frac{p^2 + 8\pi \aN - \sqrt{|p|^4 + 16 \pi \aN p^2}}{\sqrt{|p|^4 + 16 \pi \aN p^2}} \, , \end{split} \]
and the commutator
\[ \big[ \frac{a_0^\dagger}{\sqrt{N}} a_{-p} , a_{-p}^\dagger \frac{a_0}{\sqrt{N}} \big] = \frac{1}{N} a_0^\dagger a_0 - \frac{1}{N} a_{-p}^\dagger a_{-p} = 1 - \frac{1}{N} (\cN_+ + a_{-p}^\dagger a_{-p}) \, . \]
The contribution proportional to $N^{-1}$ on the r.h.s. of the last equation produces (using Lemma \ref{lm:tau}) the error $\delta$. 
Inserting in (\ref{eq:diag0}), we conclude that 
\begin{equation}\label{eq:diag1}
\begin{split} 
&e^{-\B_3} e^{-\B_2} H_N  e^{\B_2} e^{\B_3}
	\\ &= 4 \pi \aN (N-1) - \frac{1}{2} \sum_{|p| \leq N^\alpha} \Big[ p^2 + 8\pi \aN - \sqrt{|p|^4 + 16\pi \aN p^2} - \frac{(8 \pi \aN)^2}{2p^2} \Big] 	\\ &\quad +  \sum_{|p| \leq N^\alpha} \sqrt{|p|^4 + 16 \pi \aN p^2} \, \big(\gamma_p b_p^\dagger + \nu_p b_{-p} \big) \big( \gamma_p b_p + \nu_p b_{-p}^\dagger \big) + \sum_{|p| > N^\alpha} p^2 a_p^\dagger a_p + Q_4 
 + \mathcal{E} ,
\end{split} 
\end{equation}
where $\mathcal{E}$ still satisfies the estimate (\ref{eq:EB3}) and where we introduced the modified creation and annihilation operators  
\begin{equation}\label{eq:def-b} b_p = \frac{a_0^\dagger}{\sqrt{N}} a_p, \quad b_p^\dagger = a_p^\dagger \frac{a_0}{\sqrt{N}} \end{equation} 
satisfying the commutation relations 
\begin{equation}\label{eq:b-comm} \big[ b_p, b_q\big] = \big[ b_p^\dagger , b_q^\dagger \big] = 0 , \qquad \big[ b_p , b_q^\dagger \big] = \delta_{p,q} \big( 1- \frac{\cN_+}{N} \big) - \frac{1}{N} a_q^\dagger a_p \end{equation} 
and $[ a_p^\dagger a_r , b^\dagger_q ] = \delta_{r,q} b_p^\dagger$, $[a_p^\dagger a_r, b_q ] = - \delta_{p,q} b_r$. On states with few excitations $a_0, a_0^\dagger \simeq \sqrt{N}$ we have $b_p^\dagger \simeq a_p^\dagger$, $b_p \simeq a_p$. According to Bogoliubov theory, we can therefore expect that the operators $(\gamma_p b_p^\dagger. + \nu_p b_{-p})$ and $(\gamma_p b_p + \nu_p b_{-p}^\dagger)$ can be rotated back to $b_p^\dagger$ and, respectively, $b_p$, through conjugation of the Hamiltonian with the unitary transformation generated by the antisymmetric operator  
\[ \B_4 =\frac{1}{2} \sum_{|p| \leq N^\alpha} \tau_p \big(  b_p^\dagger b_{-p}^\dagger - b_p b_{-p} \big) = \frac{1}{2} \sum_{|p| \leq N^\alpha} \tau_p \big( a_p^\dagger a_{-p}^\dagger \frac{a_0 a_0}{N} - \text{h.c.} \big)  .   \]
Notice that $\B_4$ has the same form as the operator $\B_2$ defined in (\ref{defB}) (with a different choice of the coefficients, of course; here it is more convenient to keep the factor $N^{-1}$ out of $\tau_p$). To control the action of $\B_4$, we will need rough a-priori bounds on the growth of the number and the energy of the excitations.
\begin{lemma} \label{lm:gronB4} 
For every $k \in \mathbb{N}$, we have \footnote{The estimate for $Q_4$ will only be used in the next section, to show upper bounds on the eigenvalues of $H_N$; for the lower bounds, we will only use the fact that $Q_4 \geq 0$.} 
\begin{align}
e^{-\B_4} (\N_++1)^k e^{\B_4} &\lesssim   (\N_++1)^k, \label{eq:B4_N}
\end{align} 
Moreover
\begin{align} 
e^{-\B_4} H_1 e^{\B_4} &\lesssim H_1 + N^{\alpha}, \label{eq:B4_H1} \\
e^{-\B_4} Q_4 e^{\B_4} &\lesssim Q_4 + N^{-1+2\alpha} + N^{-1} (\cN_+ + 1)^2.  \label{eq:B4_Q4}
\end{align}
\end{lemma}
\begin{proof}  
The proof of (\ref{eq:B4_N}) is standard (based on Gronwall's lemma and on the bounds in Lemma \ref{lm:tau}). To prove \eqref{eq:B4_H1}, we define $g(s) =  e^{-s\B_4} H_1 e^{s\B_4}$ and we compute (using the commutation relations after (\ref{eq:b-comm})) 
\begin{align*}
g'(s) &= e^{-s\B_4} [H_1,\B_4] e^{s\B_4} = \sum_{|p| \leq N^\alpha} p^2 \tau_p e^{-s\B_4}  \big[ b^\dagger_p b^\dagger_{-p} + \text{h.c.} \big] e^{s\B_4} \lesssim e^{-s \B_4} H_1 e^{s \B_4} + \sum_{|p| \leq N^\alpha}  p^2 \tau_p^2. 
\end{align*}
From Lemma \ref{lm:tau}, we have $\sum_p p^2 \tau_p^2 \lesssim N^\alpha$; with Gronwall's lemma, we obtain (\ref{eq:B4_H1}). 

In order to show (\ref{eq:B4_Q4}), we set $h(s) = e^{-s\B_4} Q_4 e^{s\B_4}$, then 
\begin{equation}\label{eq:hder} h' (s) = e^{-s\B_4} [Q_4, \B_4] e^{s\B_4}. \end{equation} 
Proceeding as in (\ref{eq:Q4B2}), we find (we use here the convention that $\tau_q = 0$, for $|q| > N^\alpha$) 
\[ \begin{split} [Q_4, \B_4 ] = \; & \frac{1}{2} \sum_{q} (\hat{V}_N * \tau) (q)  b_q^\dagger b_{-q}^\dagger + \sum_{p,q,s} \hat{V}_N (q-p) \tau_s b^\dagger_{p+s-q} b^\dagger_q a^\dagger_{-s} a_p +\text{h.c.} \end{split} \]
Switching to position space, we write  
\[  [Q_4, \B_4 ] = \frac{1}{2} \int_{\Lambda^2} \dx  \dy  V_N (x-y) \check{\tau} (x-y) b_x^\dagger b_y^\dagger + \int_{\Lambda^3} \dx  \dy  \dz  V_N (x-y) \check{\tau} (x-z) b_x^\dagger b_y^\dagger a_z^\dagger a_y + \text{h.c.} \]
With the bounds from Lemma \ref{lm:tau}, we conclude 
\[ \pm [Q_4, \B_4] \lesssim Q_4 + N^{-1+2\alpha} + N^{-1} (\cN_+ + 1)^2. \]
Inserting this in (\ref{eq:hder}), applying (\ref{eq:B4_N}) and then Gronwall's lemma, we obtain (\ref{eq:B4_Q4}). 
\end{proof} 

We are now ready to state the main result of this section, which shows that conjugation with $e^{\cB_4}$ diagonalize the quadratic part of the Hamilton operator. 
\begin{proposition}\label{prop:B4}
We have
\begin{equation}\label{eq:propB4-cl}
\begin{split}
e^{-\B_4} e^{-\B_3} &e^{-\B_2} H_N e^{\B_2} e^{\B_3}e^{\B_4} \\ = &\; 4 \pi \aN (N-1)  + \frac{1}{2}\sum_{p} \left[ \sqrt{p^4 + 16 \pi \aN p^2} - p^2 - 8 \pi \aN + \frac{(8\pi \aN)^2}{2p^2}\right] \\ &+ \sum_{p} \sqrt{p^4 + 16 \pi \aN p^2} \, a^\dagger_p a_p  + e^{-\B_4}Q_4e^{\B_4} + \mathcal E_{\B_4}
\end{split}
\end{equation} 
where 
\begin{equation}\label{eq:cE4}\begin{split} \pm \mathcal{E}_{\B_4}  \lesssim  \; &N^{-3\alpha/2} (H_1 +N^\alpha) + N^{-\alpha/2} e^{-\B_4} Q_4 e^{\B_4} + N^{-\alpha/2} (\cN_+ + 1)  \\ &+ N^{(3\alpha-1)/2} (\cN_+ + 1)^{3/2} + N^{-1+5\alpha/2} (\cN_+ + 1)^2.   \end{split} \end{equation} 
\end{proposition}

\begin{proof}
For $s \in [0;1]$, we define
\[ E(s) = \sum_{|p| \leq N^\alpha} \sqrt{|p|^4 + 16 \pi \aN p^2} \,  \big(\gamma^{s}_p b_p^\dagger  + \nu_p^{s} b_{-p} \big) \big( \gamma_p^{s}  b_p + \nu^{s}_p b_{-p}^\dagger  \big) \]
with the operators $b_p, b^\dagger_p$ defined in (\ref{eq:def-b}) and with the notation $\gamma_p^{s} = \cosh (s \tau_p)$ and $\nu_p^{s} = \sinh (s \tau_p)$. In particular, for $s=1$, this is exactly the operator appearing in (\ref{eq:diag1}), on the third line of the equation.  For $\psi$ in the sector $\{ \cN = N \}$, we define $f_\psi :[0;1] \to \bR$ by 
\[ f_\psi (s) = \langle \psi , e^{-s\B_4} E(s) e^{s\B_4} \psi \rangle \, . \]
The idea is that the generalized Bogoliubov transformation $e^{s\B_4}$ approximately cancels (on states with few excitations) the symplectic rotations determined by the coefficients $\gamma_p^{s}$, $\nu_p^{s}$ (it would precisely cancel it, if the operators $b_p^\dagger, b_p$ satisfied canonical commutation relations); hence, on states with few excitations, we expect $f_\psi$ to be approximately constant in $s$. More precisely, we claim that 
\begin{equation}\label{eq:claim-diag} 
\big| f'_\psi (s) \big| \lesssim N^{2\alpha-1} \langle \psi, (\cN_+ + 1)^2 \psi \rangle \, .
\end{equation} 
Assuming for a moment (\ref{eq:claim-diag}) to hold true, we could conclude, integrating over $s \in [0;1]$, that 
\[ e^{-\B_4} E(1) e^{\B_4} = E(0) + \delta \]
with 
\[ \pm \delta \lesssim N^{2\alpha-1} (\cN_+ + 1)^2 \, . \]
With the bounds from Lemma \ref{lm:gronB4} (and noticing that the action of $\B_4$ on the high-momenta part of the kinetic energy is trivial), this would imply that   
\begin{equation}\label{eq:diag2} \begin{split} e^{-\B_4} &e^{-\B_3} e^{-\B_2} H_N e^{\B_2} e^{\B_3} e^{\B_4} \\ &= 4 \pi \aN (N-1) - \frac{1}{2} \sum_{|p| \leq N^\alpha} \Big[ p^2 + 8\pi \aN - \sqrt{|p|^4 + 16\pi \aN p^2} - \frac{(8 \pi \aN)^2}{2p^2} \Big] 	\\ &\quad + \sum_{|p| \leq N^\alpha} \sqrt{|p|^4 + 16 \pi \aN p^2} \, a_p^\dagger \frac{a_0 a_0^\dagger}{N} a_p + \sum_{|p| > N^\alpha} p^2 a_p^\dagger a_p + e^{-\B_4} Q_4 e^{\B_4} + \mathcal{E} ,
\end{split} \end{equation} 
where $\mathcal{E}$ satisfies (\ref{eq:cE4}). Writing $a_0 a_0^\dagger = a_0^\dagger a_0 + 1 = N - \cN_+ + 1$, we could then replace 
\[  \sum_{|p| \leq N^\alpha} \sqrt{|p|^4 + 16 \pi \aN p^2} \, a_p^\dagger \frac{a_0 a_0^\dagger}{N} a_p =  \sum_{|p| \leq N^\alpha} \sqrt{|p|^4 + 16 \pi \aN p^2} \, a_p^\dagger a_p + \delta \]
with 
\[ \pm \delta \leq \frac{1}{N} \sum_{|p| \leq N^\alpha} \sqrt{|p|^4 + 16 \pi \aN p^2}  a_p^\dagger (\cN_+ + 1) a_p \lesssim N^{-1+2\alpha} (\cN_++1)^2\, . \]
Furthermore, since 
 \begin{align*}
\Big| \, p^2 + 8 \pi \aN - \sqrt{|p|^4 + 16 \pi \aN p^2} -  \frac{(8\pi \aN)^2}{2p^2} \Big| \lesssim (1+p^2)^{-2}
 \end{align*}
we could write 
\begin{multline*}
\sum_{|p|\leq N^{\alpha}} \left[ p^2 + 8\pi \mathfrak{a}_N - \sqrt{|p|^4 + 16 \pi \aN p^2} -  \frac{(8\pi \aN)^2}{2p^2}\right] \\ = \sum_{p} \left[p^2 + 8\pi \mathfrak{a}_N - \sqrt{|p|^4 + 16 \pi \aN p^2} -  \frac{(8\pi \aN)^2}{2p^2}\right] + \mathcal O(N^{-\alpha}).
\end{multline*}
Similarly, from $\big| p^2 - \sqrt{|p|^4 + 16 \pi \aN p^2} \big| \lesssim 1$, we could bound 
\[ \pm \sum_{|p| > N^\alpha} \big[ p^2 - \sqrt{|p|^4 + 16 \pi \aN p^2} \big] a_p^\dagger a_p \lesssim N^{-2\alpha} H_1\, . \]
Inserting all these estimates in (\ref{eq:diag2}), we would end up with (\ref{eq:propB4-cl}) and (\ref{eq:cE4}). 

It remains to show the bound (\ref{eq:claim-diag}). To this end, we observe that 
\begin{equation}\label{eq:f's} f'_\psi (s) = \frac{d}{ds} \big\langle \psi, e^{s \B_4} E(s) e^{-s\B_4} \psi \big\rangle = \big\langle \psi, e^{s \B_4} \big\{ \big[ \B_4, E(s) \big] + \frac{\partial E(s)}{\partial s} \big\} e^{-s\B_4} \psi \big\rangle  \end{equation} 
We have, denoting $\eps_p  = \sqrt{|p|^4 + 16 \pi \aN p^2}$, 
\[ \begin{split}  \big[ \B_4, E(s) \big] = \frac{1}{2} \sum_{|p| \leq N^\alpha} \eps_p \sum_{|q| \leq N^\alpha} \tau_q 
\big\{ &\big[ b_q^\dagger b_{-q}^\dagger, (\gamma_p^s b_p^\dagger + \nu_p^s b_{-p}) \big] (\gamma_p^s b_p + \nu_p^s b_{-p}^\dagger) \\ &+ (\gamma_p^s b_p^\dagger + \nu_p^s b_{-p}) \big[ b_q^\dagger b_{-q}^\dagger, (\gamma_p^s b_p + \nu_p^s b_{-p}^\dagger) \big] \big\} + \text{h.c.} \end{split} \]
A long but straightforward computation, based on the commutation relations (\ref{eq:b-comm}), leads to 
\begin{equation}\label{eq:der-E} \begin{split} \big[ &\B_4, E(s) \big] \\ = &\, -\frac{1}{2} \sum_{|p| \leq N^\alpha} \eps_p \tau_p \nu_p^s b_p^\dagger \big(1-\frac{\cN_+}{N}\big) (\gamma_p^s b_p + \nu_p^s b_{-p}^\dagger ) - \frac{1}{2} \sum_{|p| \leq N^\alpha} \eps_p \tau_p \nu_p^s \big(1-\frac{\cN_+}{N} \big) b_p^\dagger  (\gamma_p^s b_p + \nu_p^s b_{-p}^\dagger ) \\  
 &-\frac{1}{2} \sum_{|p| \leq N^\alpha} \eps_p \tau_p \gamma_p^s (\gamma_p^s b^\dagger_p + \nu_p^s b_{-p} )
 b_{-p}^\dagger \big(1-\frac{\cN_+}{N}\big)  -\frac{1}{2} \sum_{|p| \leq N^\alpha} \eps_p \tau_p \gamma_p^s (\gamma_p^s b^\dagger_p + \nu_p^s b_{-p} )\big(1-\frac{\cN_+}{N}\big)  b_{-p}^\dagger  \\
 &+ \frac{1}{2N}  \sum_{|p|, |q| \leq N^\alpha} \eps_p \tau_q \nu_p^s \big[ b_q^\dagger a_{-q}^\dagger a_{-p} + a_q^\dagger a_{-p} b_{-q}^\dagger \big] (\gamma_p^s b_p + \nu_p^s b_{-p}^\dagger ) \\ 
  &+ \frac{1}{2N}  \sum_{|p|, |q| \leq N^\alpha} \eps_p \tau_q \gamma_p^s  (\gamma_p^s b^\dagger_p + \nu_p^s b_{-p})  \big[ b_q^\dagger a_{-q}^\dagger a_{p} + a_q^\dagger a_{p} b_{-q}^\dagger \big]
+ \text{h.c.} 
\end{split} \end{equation} 
To compute the explicit time derivative of the observable $E(s)$, on the other hand, we notice that $ d\gamma_p^s /ds = \tau_p \nu_p^s$ and that $d\nu_p^s / ds = \tau_p \gamma_p^s$. Thus, we obtain 
\[ \begin{split} 
\frac{\partial E(s)}{\partial s} = &\sum_{|p| \leq N^\alpha} \eps_p \tau_p \big(\nu_p^s b_p^\dagger + \gamma_p^s b_{-p} \big) \big( \gamma_p^s b_p + \nu_p^s b_{-p}^\dagger \big) + \sum_{|p| \leq N^\alpha} \eps_p \tau_p (\gamma_p^s b_p^\dagger + \nu_p^s b_{-p}) (\nu_p^s b_p + \gamma_p^s b_{-p}^\dagger)\, . \end{split} \]
Combining the last equation with (\ref{eq:der-E}), we observe that (as expected) all large contributions cancel. We find 
\[  \begin{split} \big[ \B_4, E(s) \big] + \frac{\partial E(s)}{\partial s} = \; &\frac{1}{N} \sum_{|p| \leq N^\alpha} \eps_p \tau_p \nu_p^s  \, \cN_+ b_p^\dagger (\gamma_p^s b_p + \nu_p^s b_{-p}^\dagger) \\ 
&+\frac{1}{N} \sum_{|p| \leq N^\alpha} \eps_p \tau_p \gamma_p^s \, (\gamma_p^s b^\dagger_p + \nu_p^s b_{-p}) \cN_+  b_{-p}^\dagger \\
&+\frac{1}{N} \sum_{|p|, |q| \leq N^\alpha} \eps_p \tau_q \nu_p^s \, b_q^\dagger a_{-q}^\dagger a_{-p} (\gamma_p^s b_p + \nu_p^s b_{-p}^\dagger) \\
&+ \frac{1}{N} \sum_{|p|, |q| \leq N^\alpha} \eps_p \tau_q \gamma_p^s \, (\gamma_p^s b_p^\dagger + \nu_p^s b_{-p}) b_q^\dagger a_{-q}^\dagger a_p + \text{h.c.} \end{split} \]
%
%\left[ (\nu_p^s)^2 \cN_+ + \gamma_p^s)^2 (\cN_+ - 1)
%\right] b_p^\dagger b_{-p}^\dagger \\
%+ \frac{1}{N} \sum_{|p|, |q| \leq N^\alpha} \eps_p \tau_q \nu_p^{s} \gamma_p^s \left[ b_q^\dagger a_{-q}^\dagger %a_{-p} b_p  + b_{-p} b_q^\dagger a_{-q}^\dagger a_p \right] \\
%+ \frac{2}{N} \sum_{|p|, |q| \leq N^\alpha} \eps_p \tau_p \left[ (\nu_p^s)^2 + (\gamma_p^s)^2 \right] b_q^\dagger %a_{-q}^\dagger b_{-p}^\dagger a_{-p} +\frac{1}{N} \sum_{|p|,|q| \leq N^\alpha} \eps_p \tau_q (\sigma_p^s)^2 %b_q^\dagger b_{-q}^\dagger \end{split} \]
Using the bounds in Lemma \ref{lm:tau}, with the estimate $\eps_p \lesssim p^2$ and the restrictions $|p|, |q| \leq N^\alpha$, we arrive at 
\[ \pm \left\{  \big[ \B_4, E(s) \big] + \frac{\partial E(s)}{\partial s}  \right\} \lesssim N^{2\alpha-1} (\cN_+ + 1)^2 \,. \]
From (\ref{eq:f's}), this implies that 
\[ |f'_\psi (s)| \lesssim N^{2\alpha-1} \langle \psi, e^{-s\B_4} (\cN_+ + 1)^2 e^{s\B_4} \psi \rangle \, . \] 
Applying Lemma \ref{lm:gronB4}, we obtain (\ref{eq:claim-diag}). 
\end{proof}

\section{Optimal BEC and Proof of Theorem \ref{thm:main}}

Let us denote
\begin{equation}\label{eq:wtHN}
\widetilde H_N = H_N -  4 \pi \aN (N-1) + \frac{1}{2}\sum_{p} \left[ \sqrt{|p|^4 + 16 \pi \aN p^2} - p^2 - 8 \pi \aN +  \frac{(8\pi \aN)^2}{2p^2}\right],
\end{equation}
and 
\begin{equation}\label{eq:Einf}
E_{\infty} = \sum_{p} \sqrt{|p|^4 + p^2 16 \pi \aN } \, a_p^\dagger a_p \,. 
\end{equation}
Moreover, let $\mathcal{U} = e^{\B_2} e^{\B_3} e^{\B_4}$. Observe that $\mathcal{U}$ is a unitary operator. From Prop. \ref{prop:B4}, we have 
\begin{equation}\label{eq:proB42}
\mathcal U^\dagger \widetilde H_N  \mathcal U = E_\infty  + e^{-\B_4}Q_4e^{\B_4} +  \mathcal E_{\B_4} , 
\end{equation}
where $\mathcal{E}_{\B_4}$ satisfies the bound (\ref{eq:cE4}). To prove that the error term $\mathcal{E}_{\B_4}$ is small, we show first that low-energy states exhibit complete Bose-Einstein condensation. 
\begin{proposition}[Optimal BEC]
	\label{prop:BEC}
On $\{\N = N\}$, we have
\begin{equation}\label{eq:BEC-op}
H_N \geq 4\pi \aN N + C^{-1} \N_+ - C,
\end{equation} 
for some constant $C>0$ independent of $N$.
\end{proposition}
\begin{proof} 
To take care of the terms on the second line of (\ref{eq:cE4}), we use localization in the number of particles, a tool developed in \cite{LewNamSerSol-13} and, in the present setting, in \cite{BBCS4}. We make use, here, of the results of \cite{LieSei-02,LieSei-06,NamRouSei-15}, which imply that, if $\psi_N \in L^2_s (\Lambda^N)$ is a normalized sequence of approximate ground states of the Hamilton operator $H_N$, satisfying  
\[ \Big| \frac{1}{N} \langle \psi_N, H_N \psi_N \rangle - 4\pi \mathfrak{a}_N \Big| \to 0 \]
as $N \to \infty$, then $\psi_N$ exhibit condensation, in the sense that 
\begin{equation}\label{eq:BEC_input}  \lim_{N \to \infty}  \frac{1}{N} \langle \psi_N , \cN_+ \psi_N \rangle = 0. \end{equation} 
Let now $f,g : \mathbb{R} \to [0,1]$ be smooth functions, such that $f(s)^2 + g(s)^2 = 1$ for all $s \in \mathbb{R}$, and $f(s) = 1$ for $s \leq 1/2$, $f(s) =0$ for $s\geq 1$. For $M_0 \geq 1$, we define $f_{M_0}(\N_+) = f(\N_+ / M_0)$ and $g_{M_0}(\N_+) = g(\N_+/M_0)$. Then, we have
\begin{align}
H_N &= f_{M_0} H_N f_{M_0} + g_{M_0} H_N g_{M_0} + \mathcal E_{M_0}, 
\label{eq:loc_identity}
\end{align} 
with 
\begin{equation} 
\mathcal E_{M_0} = \frac{1}{2} \left( [f_{M_0},[f_{M_0},H_N]] +  [g_{M_0},[g_{M_0},H_N]]\right).\label{eq:error_loc}
\end{equation}
In view of \eqref{eq:decompo_HN}, we can write (with $h =f,g$), 
\begin{align*}
[h_{M_0} , [h_{M_0},H_N]] 
	&=  \big[ h (\N_+ / M_0) - h ((\N_+ -2)/M_0)\big]^2 \sum_{p \neq 0} \hat V_N(p) a_p^\dagger a_{-p}^\dagger a_0 a_0 + \text{h.c.}, \\
	&\quad +  \big[ h (\N_+ / M_0) - h ((\N_+ -1)/M_0)\big]^2  \sum_{q,r,q+r\neq 0} \hat V_N(r) a_{q+r}^\dagger a_{-r}^\dagger a_{q}a_0 + \text{h.c.},
\end{align*}
This easily implies that  
\begin{align}
	\label{eq:local_error}
\pm \mathcal E_{M_0} \lesssim M_0^{-2} (Q_4 + N) \1^{\{M_0/3 \leq \N_+ \leq 2M_0\}}.
\end{align}

We choose $M_0 = \varepsilon N$, for some $\varepsilon >0$, independent of $N$, to be fixed later on. We introduce the notation $\N_{+}^{\, \mathcal U} = \mathcal U^\dagger \N_+ \mathcal U$. We use (\ref{eq:proB42}) with the bound (\ref{eq:cE4}) for the error term $\mathcal{E}_4$; we pick $\alpha = - \log \ell / \log N$, so that $N^\alpha = 1/\ell$, for some $\ell > 0$, independent of $N$, to be specified below. For $N$ large enough, we obtain from Prop. \ref{prop:B4},
\begin{align}
f_{M_0}(\N_+) H_N  &f_{M_0}(\N_+) 
	=   \mathcal U f_{M_0}(\N_{+}^{\, \mathcal U}) \mathcal U^\dagger H_N \mathcal U \, f_{M_0}(\N_{+}^{\, \mathcal U}) \mathcal U^\dagger \nonumber \\
	&\geq \mathcal U f_{M_0}(\N_{+}^{\, \mathcal U}) \left( 4\pi \aN N + H_1 -C  + e^{-\B_4}Q_4e^{\B_4} + \mathcal E_4  \right)f_{M_0}(\N_{+}^{\, \mathcal U}) \mathcal U^\dagger\nonumber  \\
	&\geq \mathcal U f_{M_0}(\N_{+}^{\, \mathcal U}) \big( 4\pi \aN N + (1 - C \ell^{1/2} - C\ell^{-5/2} N^{-1}- C \ell^{-7/2}\eps) (H_1+1) \nonumber \\
	&\qquad \qquad\qquad\qquad\qquad\qquad\qquad  + (1-C \ell^{1/2})  e^{-\B_4}Q_4e^{\B_4}\big)f_{M_0}(\N_{+}^{\, \mathcal U}) \mathcal U^\dagger \nonumber \\
	&\geq f_{M_0}(\N_{+})^2 (4\pi \aN N - C + C^{-1} \N_+), \label{eq:BEC_f}
\end{align}
choosing first $\ell > 0$ small enough and then $\eps > 0$ sufficiently small. Here, we used $(\N_++1)^j \lesssim (\N_+^{\, \mathcal U}+1)^j$ (as follows from Lemma \ref{Gron}, Lemma \ref{lemma:gronwall_B3} and Lemma \ref{lm:gronB4}), to estimate the error terms on the second line of (\ref{eq:cE4}). Moreover, we bounded $\cN_+ , \cN^{\mathcal U}_+ \lesssim H_1$. 

On the other hand, following an argument from \cite[Prop. 6.1]{BBCS4}, we find 
\begin{align} \label{eq:BEC_g}
g_{M_0}(\N_+) \left(H_N -4\pi \aN N\right)  &g_{M_0}(\N_+) \geq C^{-1} \N_+ g_{M_0}(\N_+)^2.
\end{align}
Indeed, otherwise we could find a normalized sequence $\Psi_N$, supported on $\{\N_+ > \varepsilon N\}$, satisfying  
\[\Big|  \frac{1}{N} \langle \psi_N , H_N \psi_N \rangle_{\Psi_N} - 4\pi \aN \Big| \to 0  \]
as $N \to \infty$, in contradiction with \eqref{eq:BEC_input}.

Finally, we deal with the error term $\mathcal E_{M_0=\varepsilon N}$. For $\psi_N$ with $\langle \psi_N, H_N \psi_N \rangle \leq C N$, we immediately find, from \eqref{eq:local_error}, that 
\[ \langle \psi_N, \mathcal E_{M_0 = \varepsilon N} \psi_N \rangle \lesssim \varepsilon^{-2} N^{-1}. \]
Since (\ref{eq:BEC-op}) holds trivially, on states with $\langle \psi_N , H_N \psi_N \rangle \geq CN$, this, together with (\ref{eq:BEC_f}) and (\ref{eq:BEC_g}) concludes the proof of  Prop. \ref{prop:BEC}.
\end{proof}

With Prop. \ref{prop:B4} and Prop. \ref{prop:BEC}, we are now ready to show Theorem \ref{thm:main}, determining the low-energy spectrum of the operator Hamilton operator $H_N$.

\begin{proof}[Proof of Theorem \ref{thm:main}]
We continue to use the notation $\widetilde{H}_N$ and $E_\infty$ introduced in (\ref{eq:wtHN}), (\ref{eq:Einf}). Moreover, we denote by $\lambda_1 (\widetilde{H}_N) \leq \lambda_2 (\widetilde{H}_N) \leq \dots$ and by $\lambda_1 (E_\infty) \leq \lambda_2 (E_\infty) \leq \dots$ the ordered eigenvalues of $\widetilde{H}_N$ and, respectively, of $E_\infty$. We choose now $L \in \mathbb{N}$, with $\lambda_L (\widetilde H_N) \leq \Theta$ for some $1 \leq \Theta  \leq N^{1/17}$. Then, we claim that 
\begin{equation}\label{eq:claim} \lambda_L (\widetilde H_N) = \lambda_L(E_\infty) + \mathcal O (\Theta N^{-1/17}). \end{equation} 

Since $\lambda_0 (E_\infty) = 0$, (\ref{eq:claim}) shows that the ground state energy $E_N$ of $H_N$ satisfies (\ref{eq:ENbd}). 
%\[ \Big| E_N - \Big\{  4\pi \mathfrak{a}_N (N-1)  - \frac{1}{2}\sum_{p} \left[ \sqrt{|p|^4 + 16 \pi 
%\aN p^2} - p^2 - 8 \pi \aN-  \frac{(8\pi \aN)^2}{2p^2}\right] \Big\} \Big| \lesssim N^{-1/17} \]
It is then easy to check, using (\ref{eq:claim}) that the excitations of $H_N - E_N$ satisfy the claim (\ref{eq:excit}). 

To prove (\ref{eq:claim}), we first show a lower bound and then a matching upper bound. 
We use again Prop. \ref{prop:B4}, but this time we choose the exponents $\alpha = 2/17$.

{\it Lower bound on $\lambda_L (\widetilde{H}_N)$.} We use again the localization identity \eqref{eq:loc_identity} but this time we take $M_0 = N^{1/2 + 1/34}$. Let $Y$ denote the subspace generated by the first $L$ eigenfunctions of $\widetilde H_N$ and let us denote $Z = \mathcal U^\dagger Y$ which is of dimension $L$. From of \eqref{eq:loc_identity}, we have
\begin{equation}\label{eq:low1}
\lambda_{L}(\widetilde H_N) 
	\geq P_Y \left(f_{M_0}(\N_+) \widetilde H_N f_{M_0}(\N_+) + g_{M_0}(\N_+)\widetilde H_N g_{M_0}(\N_+) + \mathcal E_{M_0}\right) P_Y.
\end{equation}
From Prop. \ref{prop:BEC}, we have
\begin{align*}
 g_{M_0}(\N_+) \widetilde H_N g_{M_0}(\N_+) \geq C g^2_{M_0}(\N_+) (C^{-1} M_0 - C) \geq 0
\end{align*}
for $N$ large enough (recall the choice $M_0 = N^{1/2+1/34}$). Here, we used that $g_{M_0}$ is supported on $\cN_+ > M_0/3$. Moreover, with (\ref{eq:local_error}), we find 
\begin{align*}
P_Y \mathcal E_{M_0} P_Y &\geq - C M_0^{-2} P_Y (Q_4 + N) P_Y \geq  - C M_0^{-2} P_Y (H_N + N) P_Y  \geq - C M_0^2 N \geq - C N^{-1/17} 
\end{align*}
because (from the upper bound) we know that $H_N \leq C N$ on $Y$. From (\ref{eq:low1}), we obtain 
\begin{align*}
\lambda_{L}(\widetilde H_N) \geq P_Y f_{M_0}(\N_+) \widetilde H_N f_{M_0}(\N_+) P_Y - CN^{-1/17}.
\end{align*}
We now use Prop. \ref{prop:B4} to estimate 
\begin{align*}
P_Y f_{M_0}(\N_+)\, &\widetilde H_N \,f_{M_0}(\N_+) P_Y \\
	& \geq \mathcal U P_Z f_{M_0}(\N_+^{\mathcal U}) \,\mathds{1}^{\{\mathcal N_+ \leq N\}} \big( e^{\B_4}  (E_\infty + \mathcal E_4) e^{-\B_4} + Q_4 \big)\mathds{1}^{\{\mathcal N_+ \leq N\}}  \, f_{M_0}(\N_+^{\mathcal U}) P_Z \mathcal U^\dagger.
\end{align*}
Using that $\N_+ \lesssim H_1 \leq E_\infty$ and the choices $M_0 = N^{1/2 + 1/34}$, $\alpha = 2/17$, we find 
\begin{align*}
P_Y f_{M_0}(\N_+) \,&\widetilde H_N \,f_{M_0}(\N_+) P_Y \\
	&\geq \left(1- C N^{-1/17}\right) \mathcal U P_Z f_{M_0}(\N_+^{\mathcal U}) e^{\B_4} \, E_\infty \, e^{-\B_4} f_{M_0}(\N_+^{\mathcal U})P_Z \mathcal U^\dagger.
\end{align*}
Now, it turns out that for $N$ large enough $\dim  f_{M_0}(\N_+^{\mathcal U})P_Z = L$ because 
\begin{align*}
\max_{\xi \in P_Z} \frac{\bigg\|\sqrt{1-f_{M_0}(\N_+^{\mathcal U})^2}\xi\bigg\|^2}{\|\xi\|^2} \leq C M_0^{-1} \max_{\xi \in P_Y} \frac{\|\N_+^{1/2}\xi\|^2}{\|\xi\|^2} \leq C \lambda_{L}(\widetilde H_N) M_0^{-1} \underset{N \to \infty}{\longrightarrow} 0,
\end{align*}
see for instance \cite[Prop. 6.1 ii)]{LewNamSerSol-13}. Thus 
\begin{align*}
\lambda_{L}(\widetilde H_N) 
	&\geq \max_{\xi \in Y} \frac{\langle e^{-\B_4} f_{M_0}(\N_+^{\mathcal U}) P_Z \mathcal{U}^\dagger \xi , E_\infty e^{-\B_4} f_{M_0} (\N_+^{\mathcal U}) P_Z \mathcal U^\dagger\xi \rangle}{\|\xi\|^2} - C \Theta N^{-1/17} \\
	&\geq \max_{\xi \in e^{-\B_4} f_{M_0}(\N_+^{\mathcal U}) Z} \frac{\langle \xi, E_\infty \xi \rangle}{\|\xi\|^2}(1- C \lambda_{L}(\widetilde H_N) M_0^{-1}) - C \Theta N^{-1/17} \\
	&\geq \min_{\dim X = L} \max_{\xi \in X} \frac{\langle \xi , E_\infty \xi \rangle}{\|\xi\|^2}(1- C \lambda_{L}(\widetilde H_N) M_0^{-1})- C\Theta N^{-1/17} \\
	&\geq \lambda_L (E_\infty) (1 -C\Theta N^{-1/2}) -C \Theta N^{-1/17}.
\end{align*}
which implies $\lambda_L(\widetilde H_N) \geq \lambda_{L}(E_\infty) + \mathcal O(\Theta N^{-1/17})$.

\textit{Upper bound on $\lambda_L (\widetilde{H}_N)$}. Let $Z$ denote the subspace generated by the first $L$ eigenfunctions of $E_\infty$ and $P_Z$ be the orthogonal projection onto $Z$. The normalized eigenfunctions of $E_\infty$ have the form 
\begin{align}\label{eq:form_EF}
\xi = \prod_{j=1}^k \frac{a^\dagger(p_j)^{n_j}}{\sqrt{n_j!}} \Omega,
\end{align}
for some $k\geq 1$, $p_j \in \Lambda_+^*$, $n_j \geq 1$ and where $\Omega$ is the vacuum. Note that 
\begin{align*}
P_Z \N_+ P_Z  \lesssim P_Z H_1 P_Z  \leq P_Z E_\infty P_Z  \leq \lambda_{L}(E_\infty) \leq C \Theta \leq C N^{1/17},
\end{align*} 
where we used the lower bound we proved above. Since $[P_Z, \cN_+]=0$, this bound can be also applied to powers of $\cN_+$. Note also that $e^{\B_4} E_{\infty} e^{-\B_4}$ almost commutes with $\N_+$ in the sense that
\begin{align*}
\mathds{1}^{\{\mathcal N_+ \leq N\}} e^{\B_4} E_{\infty} e^{-\B_4} \mathds{1}^{\{\mathcal N_+ \leq N\}} - e^{\B_4} E_{\infty} e^{-\B_4}
	&\leq \frac{1}{2} \sum_{p} 8\pi \mathfrak{a}_N \chi_{|p| \leq N^\alpha} \big[ \mathds{1}^{\{\N_+ > N\}}a_p^\dagger a_{-p}^\dagger + a_p a_{-p}\mathds{1}^{\{\N_+ > N\}} \big] \\
	&\leq C N^{3\alpha/2-1} (\N_++1)^2.
\end{align*}
Hence, we have
\begin{align*}
P_Z e^{-\B_4} \mathds{1}^{\{\mathcal N_+ \leq N\}} e^{\B_4} E_{\infty} e^{-\B_4} \mathds{1}^{\{\mathcal N_+ \leq N\}}e^{\B_4}P_Z 
	& \leq P_Z \bigg(E_{\infty} + C N^{3\alpha/2-1} (\N_++1)^2 \bigg)P_Z \\
	& \leq \lambda_{L}(E_\infty) + C \Theta N^{-1/17}.
\end{align*}
Together with Prop. \ref{prop:B4}, we find  
\begin{equation}\label{eq:upp1}
\lambda_{L} (E_\infty) +  C \Theta N^{-1/17} \geq P_Z e^{-\B_4} \mathds{1}^{\{\N_+ \leq N\}} \left(\mathcal U^\dagger \widetilde H_N \mathcal U - e^{\B_4} \mathcal E_4  e^{-\B_4} - Q_4 \right) \mathds{1}^{\{\N_+ \leq N\}} e^{\B_4}P_Z .
\end{equation}
Again, because 
\begin{equation*}
\max_{\xi \in e^{\B_4} P_Z} \frac{\big\| \mathds{1}^{\{\N_+ > N\}} \xi\big\|^2}{\|\xi\|^2} \leq C N^{-1}\max_{\xi \in P_Z} \frac{\big\| e^{-\B_4} \N_+^{1/2} \xi\big\|^2}{\|\xi\|^2} \leq C \Theta_0 N^{-1} \underset{N \to \infty}{\longrightarrow} 0,
\end{equation*}
we have $\dim  \mathds{1}^{\{\N_+ \leq N\}}  e^{\B_4} P_Z = L$ for $N$ large enough.
 With Lemma \ref{lm:gronB4}, we obtain 
\begin{align*}
 P_Z e^{-\B_4} \mathds{1}^{\{\N_+ \leq N\}} \bigg( e^{\B_4} \mathcal E_4  e^{-\B_4} + Q_4 \bigg) \mathds{1}^{\{\N_+ \leq N\}} e^{\B_4}P_Z
 	  &\lesssim N^{-1/17} \Theta +  P_Z Q_4 P_Z. 
\end{align*}
To estimate $P_Z Q_4 P_Z$, we use an argument from \cite[Lemma 6.1]{BBCS}: from $P_Z E_\infty P_Z \leq \Theta$, we must have $a_p \xi = 0$ for all $|p| > \Theta^{1/2}$ and $\xi \in Z$. This implies that 
\begin{align*}
\langle \xi, Q_4 \xi \rangle 
	&\leq \sum_{p,q,r} \hat V_N(r) \chi_{|r|\leq \Theta^{1/2}} \|a_{p+r} a_q \xi \| \|a_{p} a_{q+r} \xi \| \\
	& \leq C \Theta^{3/2} N^{-1} \|(\N_++1)\xi\|^2 \leq C \Theta^{7/2} N^{-1} \|\xi\|^2 \leq C N^{-1/17} \|\xi\|^2,
\end{align*}
for all $\xi \in Z$. Applying the min-max principle, we conclude from (\ref{eq:upp1}) that 
\begin{align*}
\lambda_{L} (E_\infty) 
	&\geq  \max_{\xi \in \mathcal U \mathds{1}^{\{\N_+ \leq N\}}  e^{\B_4} Z} \frac{ \langle \xi, \widetilde H_N \xi\rangle}{\|\xi\|^2} - C \Theta N^{-1/17} \\
	&\geq \min_{\dim X = L} \max_{\xi \in X} \frac{ \langle \xi, \tilde  H_N \xi\rangle}{\|\xi\|^2} - C \Theta N^{-1/17} \\
	&\geq \lambda_L (\widetilde H_N) - C \Theta N^{-1/17}.
\end{align*}

\end{proof}

\subsection*{Acknowledgements} We thank P.T. Nam for insightful discussions. B.S. gratefully acknowledge support from the European Research Council through the ERC Advanced Grant CLaQS. Additionally, B. S. acknowledges partial support from the NCCR SwissMAP and from the Swiss National Science Foundation through the Grant ``Dynamical and energetic properties of Bose-Einstein condensates''.


\begin{thebibliography}{55}
 
\bibitem{ABS}
A.~Adhikari, C.~Brennecke, B.~Schlein.
Bose-Einstein Condensation Beyond the Gross-Pitaevskii Regime.
{\it Ann. Henri Poincar\'e } {\bf 22}, (2021) 1163--1233.

%\bibitem{BDS}
%N.~{Benedikter}, G.~{de Oliveira} and B.~{Schlein}. {Quantitative derivation of the Gross-Pitaevskii equation}. \emph{Comm. Pure Appl. Math.} (2014).
%\bibitem{BFKT}
%T. Balaban, J. Feldman, H. Kn\"orrer, E. Trubowitz. Complex Bosonic Many-Body Models:
%Overview of the Small Field Parabolic Flow. {\it Ann. Henri Poincar\'e } {\bf 18} (2017), 2873--2903.

%\bibitem{BenPorSch-15}
%{\sc N.~{Benedikter}, M.~{Porta}, and B.~{Schlein}}, {\em {Effective Evolution
%  Equations from Quantum Dynamics}}, SpringerBriefs in Mathematical Physics,
%  2016.


\bibitem{BCOPS} 
G. Basti, S. Cenatiempo, A. Olgiati, G. Pasqualetti, B. Schlein. 
Ground state energy of a Bose gas in the Gross-Pitaevskii regime. Preprint arxiv:2202.10270. 

\bibitem{BCS} 
G. Basti, S. Cenatiempo, B. Schlein. A new second-order upper bound for the ground state energy of dilute Bose gases. {\it Forum of Mathematics, Sigma} {\bf 9} (2021) , e74.

\bibitem{BBCS0}
C. Boccato, C. Brennecke, S. Cenatiempo, B. Schlein. Complete Bose-Einstein condensation in the Gross-Pitaevskii regime. {\it Comm. Math. Phys.} {\bf 359} (2018), no. 3, 975--1026.

\bibitem{BBCS1}
C. Boccato, C. Brennecke, S. Cenatiempo, B. Schlein. The excitation spectrum of Bose gases interacting through singular potentials. {\it J. Eur. Math. Soc.} {\bf 22} (2020), no. 7, 2331–2403

\bibitem{BBCS}
C. Boccato, C. Brennecke, S. Cenatiempo, B. Schlein. Bogoliubov Theory in the Gross-Pitaevskii limit. \emph{Acta Math.} \textbf{222} (2019), no. 2, 219--335.

\bibitem{BBCS4}
C. Boccato, C. Brennecke, S. Cenatiempo, B. Schlein. Optimal Rate for Bose-Einstein Condensation in the Gross-Pitaevskii Regime. \emph{Comm. Math. Phys.} {\bf 376}, 1311–1395 (2020), %doi:10.1007/s00220-019-03555-9. %Preprint arXiv:1812.03086.

\bibitem{bog}
N. N. Bogoliubov. On the theory of superfluidity.
{\it Izv. Akad. Nauk. USSR} {\bf 11} (1947), 77. Engl. Transl. {\it J. Phys. (USSR)} {\bf 11} (1947), 23.

\bibitem{BCaS}
C. Brennecke, M. Caporaletti, B. Schlein. Excitation spectrum for Bose gases beyond the Gross-Pitaevskii regime. Preprint arxiv:2104.13003.


\bibitem{BSS1}
 C. Brennecke, B. Schlein, S. Schraven. Bose-Einstein Condensation with Optimal Rate for Trapped Bosons in the Gross-Pitaevskii Regime. Preprint arXiv:2102.11052.

\bibitem{BSS2}
 C. Brennecke, B. Schlein, S. Schraven. Bogoliubov Theory for Trapped Bosons in the Gross-Pitaevskii Regime.  \emph{Ann. Henri Poincaré} (2022) https://doi.org/10.1007/s00023-021-01151-z



\bibitem{BS}
C. Brennecke, B. Schlein. Gross-Pitaevskii dynamics for Bose-Einstein condensates. 
 \emph{ Anal. PDE} \textbf{12} (2019), no. 6, 1513--1596.

%\bibitem{BSo}
%B. Brietzke. On the Second Order Correction to the Ground State Energy of the Dilute Bose Gas. PhD %thesis (2017).
%B. Brietzke, J.P. Solovej. The Second Order Correction to the Ground State Energy of the %Dilute Bose Gas. {\it Annales Henri Poincar\'e} {\bf 21} (2020), 571--626.
%\bibitem{BFS}
%%B. Brietzke. On the Second Order Correction to the Ground State Energy of the Dilute %Bose Gas. PhD %thesis (2017).
%B. Brietzke, S. Fournais, J.P. Solovej. A simple second order lower bound to the energy %of the dilute Bose gases. {\it Comm. Math. Phys} {\bf 376} (2020), 321--351.
%

%\bibitem{DN}
%J.~ Derezi\'nski, M.~Napi\'orkowski. Excitation Spectrum of Interacting Bosons in the Mean-Field Infinite-%Volume Limit. {\it Annales Henri Poincar\'e} {\bf
%15} (2014), 2409-2439.

%\bibitem{Dy}
%F.J. Dyson. Ground-State Energy of a Hard-Sphere Gas. {\it Phys. Rev.} {\bf 106} (1957), %20--26.

%\bibitem{ESY0}
%L.~{Erd\H{o}s}, B.~{Schlein} and H.-T.~{Yau}.
%Derivation of the Gross-Pitaevskii hierarchy for the dynamics of Bose-Einstein condensate. {\it Comm. Pure  Appl. Math.} {\bf 59} (2006), no. 12, 1659--1741.
%\bibitem{ESY}
%L. Erd\H os, B. Schlein, H.-T. Yau. Ground-state energy of a low-density Bose gas: a second order upper %bound. {\it Phys. Rev. A} {\bf 78} (2008), 053627.
%
%\bibitem{ESY2}
%L.~{Erd{\H{o}}s}, B.~{Schlein}, H.-T.~{Yau}.
%\newblock Derivation of the {G}ross-{P}itaevskii equation for the dynamics of
%  {B}ose-{E}instein condensate,
%\newblock \emph{Ann. of Math.} \textbf{172} (2010), no. 1, 291--370.

%\bibitem{EY}
%L.~{Erd{\H{o}}s} and H.-T.~{Yau}. {Derivation of the
%  nonlinear {S}chr\"{o}dinger equation from a many-body {C}oulomb system}. \emph{Adv.
%  Theor. Math. Phys.} \textbf{5} (2001),  no. 6, 1169--1205.

%\bibitem{F} S. Fournais. Length scales for BEC in the dilute Bose gas. arXiv:2011.00309.

\bibitem{FS}
S. Fournais, J.P. Solovej. The energy of dilute Bose gases. {\it Annals of Mathematics}, {\bf 192} (2020).

\bibitem{FS2} 
S. Fournais, J.P. Solovej. The energy of dilute Bose gases II: The general case. Preprint arXiv:2108.12022. 

\bibitem{Hainzl} C. Hainzl. Another Proof of BEC in the GP-limit. \emph{J. Math. Phys.}  {\bf 62} (2021), no. 5, 459--485.


\bibitem{HS} C. Hainzl, R. Seiringer. The BCS critical temperature for potentials with negative scattering length. {\it Lett. Math. Phys.}  {\bf 84} (2008)2-3, 99-107.

%\bibitem{GiuS}
%A. Giuliani, R. Seiringer. The ground state energy of the weakly interacting Bose gas at high density. {\it J. %Stat. Phys.} {\bf 135} (2009), 915.
%
%\bibitem{GS}
%P. Grech, R. Seiringer. The excitation spectrum for weakly interacting bosons in a trap. {\it Comm. Math. %Phys.} {\bf 322} (2013), no. 2, 559-591.
%
%\bibitem{LNS} M.~{Lewin}, P.~T.~{Nam} and B.~{Schlein}.
%{Fluctuations around Hartree states in the mean-field regime}. %Preprint arXiv:1307.0665.
%
%\bibitem{Lan}
%L.D. Landau. Theory of the superfluidity of Helium II.
%{\it Phys. Rev.} {\bf 60} (1941), 356�-358.
%
%\bibitem{LNR1} M. Lewin, P.~T.~Nam, N. Rougerie. Derivation of Hartree's theory for generic
%mean-field Bose gases. {\it Adv. Math.} {\bf 254} (2014), pp. 570-621.
%\bibitem{LNR2} M. Lewin, P.~T.~Nam, N. Rougerie. The mean-field approximation and the
%non-linear Schr\"odinger functional for trapped  {B}ose gases. {\it Trans. Amer. Math. Soc.} {\bf 368}
%(2016), no. 9, 6131-6157.
%\bibitem{HY} K. Huang, C. N. Yang. Quantum-Mechanical Many-Body Problem with Hard-Sphere Interaction. \emph{Phys. Rev.} \textbf{105}, 3, (1957), 767 --757.
%
%\bibitem{LHY} T. D. Lee, K. Huang, C. N. Yang. Eigenvalues and Eigenfunctions of a Bose System of Hard Spheres and Its Low-Temperature Properties. \emph{Phys. Rev.} \textbf{106} (1957), 6, 1135--1145.
%
%\bibitem{LY} T. D. Lee, C. N. Yang. Many-Body Problem in Quantum Mechanics and Quantum Statistical Mechanics. \emph{Phys. Rev.} \textbf{105} (1957), 1119 -- 1120.
%

\bibitem{LewNamSerSol-13} M.~Lewin, P.~T.~{Nam}, S.~{Serfaty}, J.P. {Solovej}. Bogoliubov spectrum of interacting Bose gases.  \newblock \emph{Comm. Pure Appl. Math.} \textbf{68} (2014), 3, 413 -- 471.

\bibitem{LieSei-02}
E.~H.~Lieb and R.~Seiringer.
\newblock Proof of {B}ose-{E}instein condensation for dilute trapped gases.
\newblock \emph{Phys. Rev. Lett.} \textbf{88} (2002), 170409.

\bibitem{LieSei-06}
E.~H.~Lieb and R.~Seiringer.
\newblock Derivation of the Gross-Pitaevskii equation for rotating Bose gases.
\newblock \emph{Comm. Math. Phys. } \textbf{264}:2 (2006), 505-537.

%\bibitem{LSSY}
%E.~H.~Lieb, R.~Seiringer, J. P. Solovej and J.~Yngvason. \textit{
%The Mathematics of the Bose Gas and its Condensation}. Series: Oberwolfach %Seminars.
%Birkh\"auser Verlag, 2005.

\bibitem{LSY}
E.~H.~Lieb, R.~Seiringer, and J.~Yngvason.
\newblock Bosons in a trap: A rigorous derivation of the {G}ross-{P}itaevskii
  energy functional. \newblock \emph{Phys. Rev. A} \textbf{61} (2000), 043602.

\bibitem{LS2}
E.~H.~Lieb and R.~Seiringer.
\newblock Derivation of the Gross-Pitaevskii Equation for Rotating Bose Gases
\newblock \emph{Comm. Math. Phys.} \textbf{264} (2006), 505--537.

%\bibitem{Sei-11}
%R.~Seiringer.
%\newblock The excitation spectrum for weakly interacting bosons
%\newblock \emph{Comm. Math. Phys.} \textbf{306} (2011), 565--578.

%\bibitem{LSo}
%E.~H.~Lieb, J. P. Solovej. Ground state energy of the one-component charged Bose gas. {\it Comm. Math. Phys.} {\bf 217} (2001), 127--163. Errata: {\it Comm. Math. Phys.} {\bf 225} (2002), 219-221.
%\bibitem{LSo2}
%E.~H.~Lieb, J. P. Solovej. Ground state energy of the two-component charged Bose gas. {\it Comm. Math. %Phys.} {\bf 252} (2004), 485--534.

\bibitem{LY}
E.~H.~Lieb, J.~Yngvason. Ground State Energy of the low density Bose Gas. {\it Phys. Rev. Lett.} {\bf 80} (1998), 2504--2507.

\bibitem{NNRT}
P.~T.~{Nam}, M. Napi{\'o}rkowski, J. Ricaud, A. Triay. Optimal rate of condensation for trapped bosons in the Gross--Pitaevskii regime. {\em Anal. PDE} (in press)

\bibitem{NamRouSei-15}
P.~T.~{Nam}, N. ~{Rougerie}, R.~Seiringer. Ground states of large bosonic systems: The Gross-Pitaevskii limit revisited.
\emph{Anal. PDE} {\bf 9} (2016), no. 2, 459--485.

\bibitem{NT} 
P.~T.~{Nam}, A. Triay. Bogoliubov excitation spectrum of trapped Bose gases in the Gross-Pitaevskii regime. Preprint arXiv:2106.11949 

\bibitem{Pachpatte}
B.~G.~Pachpatte, Inequalities for differential and integral equations. {\it Elsevier} (1997).

%\bibitem{Solovej-ESI-2014}
%{\sc J.~P. Solovej}, {\em Many body quantum mechanics}.
%\newblock Lecture notes at the Erwin Schroedinger Institute 2014, available
%  online at
%  \url{http://www.math.ku.dk/~solovej/MANYBODY/mbnotes-ptn-5-3-14.pdf}.
%\bibitem{NRS1}
%M. Napi{\'o}rkowski, R. Reuvers, J. P. Solovej. The Bogoliubov free energy functional I. Existence of %minimizers and phase diagrams. Preprint arxiv:1511.05935.
%
%\bibitem{NRS2}
%M. Napi{\'o}rkowski, R. Reuvers, J. P. Solovej. The Bogoliubov free energy functional II. The dilute limit.  %Preprint arxiv:1511.05953.
%
%
%
%\bibitem{P1}
%A. Pizzo. Bose particles in a box I. A convergent expansion of the ground state of a three-modes %Bogoliubov Hamiltonian in the mean field limiting regime. Preprint arxiv:1511.07022.
%
%\bibitem{P2}
%A. Pizzo. Bose particles in a box II. A convergent expansion of the ground state of the Bogoliubov %Hamiltonian in the mean field limiting regime. Preprint arxiv:1511.07025.
%
%\bibitem{P3}
%A. Pizzo. Bose particles in a box III. A convergent expansion of the ground state of the Hamiltonian in
%the mean field limiting regime. Preprint arxiv:1511.07026.
%
%
%\bibitem{Sei}
%R. Seiringer. The Excitation Spectrum for Weakly Interacting Bosons. {\it Comm. Math. Phys.} {\bf
%306} (2011), 565�-578.
%\bibitem{So}
%J. P. Solovej. Upper bounds to the ground state energies of the one- and two-component charged Bose %gases. {\it Comm. Math. Phys.} {\bf 266} (2006), no. 3, 797-818.
%

\bibitem{YY}
H.-T. Yau, J. Yin. The second order upper bound for the ground state energy of a Bose gas. {\it J. Stat. Phys.} {\bf 136} (2009), no. 3, 453--503.

\end{thebibliography}
\end{document}